\documentclass[acmsmall,10pt,nonacm]{acmart}
\makeatletter
\def\@ACM@checkaffil{
    \if@ACM@instpresent\else
    \ClassWarningNoLine{\@classname}{No institution present for an affiliation}%
    \fi
    \if@ACM@citypresent\else
    \ClassWarningNoLine{\@classname}{No city present for an affiliation}%
    \fi
    \if@ACM@countrypresent\else
        \ClassWarningNoLine{\@classname}{No country present for an affiliation}%
    \fi
}
\makeatother

\acmJournal{POMACS}
\setcopyright{acmcopyright}
\copyrightyear{2024}
\acmYear{2024}
\acmMonth{6}
\usepackage{graphicx}				

\usepackage{url} 
\usepackage{amsmath,amsthm}
\usepackage{bbm}
\usepackage{bm}
\usepackage{hyperref}

\usepackage{float}
\usepackage{textcomp}
\usepackage{pgfplots}
\pgfplotsset{compat=1.16} 
\usepackage{graphicx}
\usepackage{wrapfig}
\usepackage{subcaption}
\usepackage[]{youngtab}
\usepackage{tikz}
\usepackage{tikz-cd}
\usetikzlibrary{automata, positioning}
\usetikzlibrary{decorations.pathreplacing}
\usepackage{comment}
\usepackage{physics}  
\usepackage{stmaryrd}
\usepackage{theoremref}
\usepackage[toc,page]{appendix}

\usepackage{mathrsfs}

\usepackage{xspace}
\usepackage{enumitem}

\newcommand{\cdottight}{\mspace{-3mu}\cdot\mspace{-3mu}}

\newcommand{\kinfty}{\lim_{k \rightarrow \infty}}
\newcommand{\rhoone}{\lim_{\rho \rightarrow 1}}
\newcommand{\thr}{\tau}
\newcommand{\bhw}{Sub-Halfin-Whitt\xspace}
\newcommand{\shw}{Sub-Halfin-Whitt\xspace}
\newcommand{\Eps}{\mathcal{E}}

\newcommand{\snds}{Super-NDS\xspace}

\usepackage{thm-restate}
\theoremstyle{acmplain}
\newtheorem{theorem}{Theorem}
\newtheorem{lemma}[theorem]{Lemma}

\newtheorem*{claim*}{Claim}

\theoremstyle{definition}

\theoremstyle{acmdefinition}

\newcommand{\LPF}{\tiny{\mbox{LPF}}\xspace}

\newcommand{\LPFONEM}{\tiny{\mbox{LPF-1}}\xspace}
\newcommand{\LPFONE}{\mbox{LPF-1}\xspace}
\newcommand{\SERPT}{\tiny{\mbox{SERPT}}\xspace}
\newcommand{\SERPTONE}{\mbox{{\tiny SERPT-1}}\xspace}

\newcommand{\lambdak}{\lambda^{(k)}}
\newcommand{\THRESH}{\tiny{\mbox{THRESH}}\xspace}
\newcommand{\thrpolicy}{\textup{THRESH}\xspace}
\newcommand{\mand}{\qquad \mbox{and} \qquad}

\newcommand{\E}{\mathbb{E}}
\newcommand{\ET}[2][]{\mathbb{E}\left[T^{#2}_{#1}\right]}
\newcommand{\EW}[2][]{\mathbb{E}\left[W^{#2}_{#1}\right]}
\newcommand{\EN}[2][]{\mathbb{E}\left[N^{#2}_{#1}\right]}

\newcommand{\1}[1]{\mathbbm{1}_{\left(#1\right)}}
\newcommand{\const}[3]{#1^{(#2)}_{#3}}

\author[Berg]{Benjamin Berg}
\authornote{Corresponding Author}
\affiliation{%
  \institution{The University of North Carolina at Chapel Hill}
}
\email{ben@cs.unc.edu}

\author[Moseley]{Benjamin Moseley}
\affiliation{%
  \institution{Carnegie Mellon University}
}
\email{moseleyb@andrew.cmu.edu}

\author[Wang]{Weina Wang}
\affiliation{%
  \institution{Carnegie Mellon University}
}
\email{weinaw@cs.cmu.edu}

\author[Harchol-Balter]{Mor Harchol-Balter}
\affiliation{%
  \institution{Carnegie Mellon University}
}
\email{harchol@cs.cmu.edu}

\title{Asymptotically Optimal Scheduling of Multiple Parallelizable Job Classes}

\begin{document}


\begin{abstract}
Modern computing workloads are often composed of parallelizable jobs.
A parallelizable job can be completed more quickly when run on additional servers.
However, each job can only use a limited number of servers, known as its parallelizability level, which is determined by the type of computation the job performs and how it is implemented.
Workloads generally consist of multiple job classes, where jobs from different classes have different parallelizability levels and follow different job size (service requirement) distributions.

This paper considers scheduling parallelizable jobs belonging to an arbitrary number of job classes.
Given a limited number of servers, we must allocate servers across a stream of arriving jobs to minimize mean response time --- the average time from when a job arrives to the system until it completes.
We find that in lighter-load scaling regimes (i.e., Sub-Halfin-Whitt), the optimal allocation policy is Least-Parallelizable-First (LPF), which prioritizes jobs from the least parallelizable job classes regardless of their size distributions.
By contrast, we find that in the heavier-load regimes (i.e., Super-NDS), the optimal allocation policy prioritizes jobs with the Shortest Expected Remaining Processing Time (SERPT).
We also develop policies that are asymptotically optimal when the scaling regime is not known a priori.
\end{abstract}

\maketitle
\section{Introduction}
The vast majority of modern computer systems are designed to exploit parallelism.
Whether considering a single multicore server, a compute cluster of many servers, or even an entire data center, achieving good system performance requires carefully choosing how to parallelize jobs across the underlying cores or servers in the system.
For example, database queries for an in-memory database must be parallelized across the cores of a multicore server to achieve low query latency \cite{leis2014morsel, wagner2021self}, machine learning training jobs must be parallelized across the many GPUs of a compute cluster to reduce job completion time \cite{qiao2021pollux,jayaram2023sia},
and datacenter computing workloads must be parallelized across the many servers of a datacenter \cite{tumanov2016tetrisched,delgado2018kairos}.
In each of these cases, the system operator is presented with a fixed set of cores or servers and asked to devise a strategy to allocate these resources across a stream of parallelizable jobs.
However, devising optimal scheduling policies for parallelizable jobs in multiserver systems remains an open problem.

This paper develops 
asymptotically optimal scheduling policies for parallelizable jobs.
Specifically, given a limited number of servers, $k$, we 
ask how to allocate these servers to a stream of parallelizable jobs arriving over time to minimize the mean response time across jobs.

If a workload consists solely of perfectly parallelizable jobs that can each fully utilize all $k$ cores, the problem is easy because each job effectively sees one fast server that runs $k$ times faster than a single server. 
Unfortunately, in practice, parallelizable jobs typically have limits to their level of parallelizability.
For example, a database query is compiled to run with a certain number of threads, $t$, where $t \ll k$  (recall $k$ is the total number of servers) \cite{chen2016memsql}.
Allocating more than $t$ cores to this query will not reduce its processing time.
Similarly, a job's level of parallelism may be limited by overhead from communication between its threads.
For example, machine learning training jobs can be parallelized across many servers, but these servers must synchronize at the end of each training iteration \cite{lin2018model,qiao2021pollux}.
Additionally, some jobs consist of a mixture of parallelizable work and non-parallelizable work.
This leads to the well-known phenomenon described as \emph{Amdahl's Law}, which says that the speedup a job can receive from parallelism is bounded in terms of the fraction of the job that is parallelizable~\cite{Hill:2008:ALM:1449375.1449387}.

Each of these factors will affect the parallelizability of different jobs to different degrees.
Some jobs in a workload may be highly parallelizable, while others may benefit only slightly from running on additional servers.
This mixture of jobs with different levels of parallelizability is commonly seen in databases \cite{leis2014morsel}, ML workloads \cite{qiao2021pollux}, and datacenter scheduling \cite{tumanov2016tetrisched}.
\emph{This paper develops new theoretical analyses of systems with parallelizable jobs of varying parallelizability levels.}

\subsection{High-level problem statement}
To capture the limited parallelism in today's parallel computing workloads, we approximate the behavior of different parallelizable jobs by assuming that jobs are perfectly parallelizable up to some \emph{parallelizability limit}.  
This limit represents aspects of how the job was compiled, the communication overhead the job incurs, and what fraction of the job is parallelizable.  
Formally, we characterize each job by a pair $(S,c)$ where $S$ is a random variable denoting the job's {\em inherent size} and $c$ is its parallelizability level. 
The time required to complete a parallelizable job is a function of both its inherent size and the servers allocated to it over time. 
When a job runs on $m \leq c$ servers, its remaining size decreases at a rate of $m$ units of work per second until the remaining size reaches 0 and the job is considered complete.
The job cannot utilize more than $c$ servers.

To capture the varying levels of parallelizability seen within a single workload, we allow jobs to belong to one of $\ell$ \emph{job classes}.
Each job class, $i$, has an associated job size distribution, $S_i$, and parallelism limit, $c_i$, that applies to every class $i$ job. 
We assume that a job's exact size is unknown to the system, but its class is known, i.e., the system knows $(S_i, c_i)$, for all $i$.
Furthermore, we assume that jobs are \emph{malleable}, meaning that a job can change the number of servers it runs on over time \cite{gupta2014towards}.

Given a stream of malleable jobs that arrive to the system over time, a \emph{scheduling policy} must determine how many servers to allocate to each job at every moment in time.  
Given a system with $k$ homogeneous servers, our goal is to develop scheduling policies that minimize the \emph{mean response time} across jobs --- the average time from when a job arrives to the system until it is completed.

\subsection{Prior Work}\label{s:prior}

Scheduling in multiserver systems is known to be a difficult problem. 
There is a variety of work from the systems community on scheduling parallelizable jobs, including work on scheduling in databases \cite{leis2014morsel,wagner2021self}, ML clusters \cite{qiao2021pollux,jayaram2023sia}, data analytics clusters \cite{vernica2010efficient,zhu2014minimizing}, and data centers \cite{tumanov2016tetrisched,delgado2018kairos}. 
This work uses simple policies with no formal guarantees.
Additionally, many systems optimize for metrics such as fairness in addition to trying to minimize response times.

In the theoretical scheduling community, even when jobs are not parallelizable and each job occupies a single server, the optimal scheduling policy is not known.
Recent results (\cite{grosof2018srpt}) derived the first bounds on the mean response time of Shortest-Remaining-Processing-Time (SRPT) in multiserver systems, showing that SRPT is optimal in conventional heavy traffic (see Section \ref{sec:model}).
These heavy-traffic results were extended to a Gittins index policy~\cite{scully2020gittins} in the case where job sizes are not fully known. 
Unfortunately, none of these results deal with parallelizable jobs.

Another related area of active research is scheduling for \emph{multiserver jobs} \cite{harchol2022multiserver,grosof2020stability,hong2022sharp}.
Multiserver jobs each have a hard, fixed resource requirement on the number of servers they need.
Different jobs have different fixed server requirements.
In general, it is not known how to minimize the mean response time for multiserver jobs, even in heavy traffic.
Under some constraints on the server requirements of each job, an SRPT variant called ServerFilling-SRPT was recently shown to be optimal in conventional heavy traffic \cite{grosof2022wcfs,grosof2022optimal}.
Because jobs have hard resource requirements in the multiserver job model, one of the main concerns of a scheduling policy is to ensure that jobs are efficiently packed onto the available servers so that the system utilization remains high.
Failing to pack the jobs efficiently can cause the system to be unstable.
By contrast, the malleable jobs we consider do not have this stability concern, since they can adapt their levels of parallelism to ensure that the servers remain fully utilized.
Hence, the policies developed for scheduling multiserver jobs tend to prioritize packing in a way that has no benefit when scheduling malleable jobs.

The work closest to our model considers scheduling parallelizable jobs with \emph{speedup functions} \cite{berg2018,berg2020hesrpt,berg2020optimal,berg2022case,Edmonds1999SchedulingIT}.
Here, each job has an associated speedup function that determines the job's service rate on different numbers of servers.
Of these papers, \cite{berg2018,berg2020hesrpt} focus on the case where there is only one class of jobs that all follow the {\em same} speedup function.
The case of two classes of jobs is considered in  \cite{berg2020optimal,berg2022case}, but jobs are assumed to be either non-parallelizable or fully parallelizable, and the results require highly restrictive assumptions about the job size distributions. Little is known in the case where jobs have more than two levels of parallelizability. 

Some further-afield work on parallel scheduling has used different models and objectives to design scheduling policies.
Work on \emph{DAG scheduling} has considered the problem of scheduling parallelizable tasks with precedence constraints that are specified by the edges of a Directed Acyclic Graph (DAG) \cite{agrawal2016scheduling,blumofe1999scheduling,agrawal2008adaptive}.
Much of this work considers the problem of scheduling a {\em single} DAG \cite{blumofe1999scheduling,agrawal2008adaptive}.  For scheduling a stream of DAGs, only pessimistic worst-case bounds are known \cite{agrawal2016scheduling}.

\subsection{The Trade-off}
As we saw in Section~\ref{s:prior}, 
SRPT-like policies that prioritize short jobs are optimal in conventional heavy traffic for a variety of multiserver scheduling models.
However, when the system is not so heavily loaded, it is not obvious that favoring short jobs suffices for optimality.
In particular, when scheduling malleable jobs, there is another consideration that is at least as important: {\em deferring parallelizable work}.
Here, a scheduling policy can raise the long-run average utilization of the system by prioritizing less parallelizable jobs ahead of more parallelizable jobs.
The benefits of deferring parallelizable work were first described in \cite{berg2018}, which considers a less complex model with just two classes of jobs.

To see the benefit of deferring parallelizable work, suppose we have a job, A, that can only use 1 server and a job, B, that can parallelize across all $k$ servers.
If we prioritize job B, there will be a period after B completes where job A runs and $k-1$ servers sit idle.
To minimize mean response time, it is better to give 1 server to job A and use the remaining $k-1$ servers to run job B.
This holds regardless of the sizes of the jobs.
By prioritizing the less parallelizable job A, we {\em defer} working on the more parallelizable job B.

This raises the question of how to handle the {\em trade-off} between favoring jobs that are shorter versus deferring parallelizable work.
What does a policy do with a job that has a small expected remaining size, but is also highly parallelizable?  
Does the policy prioritize this job due to its small size, or defer working on this job because of its high parallelizability?
To derive an optimal scheduling policy, we must compare the relative benefits of these two effects.
We find that the optimal choice between prioritizing short jobs and deferring parallelizable work depends on the scaling behavior of the system.
When the system scales such that queueing is rare, it is more important to defer parallelizable work.
However, if the system scales such that the probability of queueing goes to 1, it is more important to prioritize short jobs.

\subsection{Contributions}
This paper derives the first results on asymptotic optimality when scheduling many classes of parallelizable jobs.
Specifically, this is the first work to develop asymptotically optimal policies when there are more than two classes of parallelizable jobs.
We analyze the mean response time of $\ell$ classes of jobs, each with a different level of parallelizability.  
We prove our results by considering how the system behaves under a variety of \emph{scaling regimes}, where the number of servers goes to infinity and the system load goes to 1 simultaneously (see Section \ref{sec:model}).
While our results rely on asymptotic analysis, the results cover a wide range of scaling regimes, including \shw regimes and mean field (constant load) scaling.
We find that the choice of an optimal policy depends on the chosen scaling regime.
Specifically, our contributions are as follows.
\begin{itemize}
\item When the system load is not too heavy (\shw), we show that the benefit of deferring parallelizable work outweighs the benefit of prioritizing short jobs.  In this case, it is beneficial to use a \emph{Least-Parallelizable-First} (LPF) policy that assigns the highest priorities to the least parallelizable classes of jobs.
We show in Theorem \ref{thm:if:opt} that LPF is asymptotically optimal with respect to mean response time.
Furthermore, we show in Theorem \ref{thm:srptsub} that it can be suboptimal to simply prioritize short jobs in these lighter-load scaling regimes.
\item Conversely, when system load is very heavy (\snds), it is more important to prioritize short jobs by using the SERPT policy as shown in Theorem \ref{thm:snds}.  In these scaling regimes, there are almost always enough jobs to fully utilize all $k$ servers.  Hence, it is more important to minimize queueing times by using SERPT than it is to defer parallelizable work.  As a result, we also prove in Theorem \ref{thm:ifsubopt} that LPF can be suboptimal in these heavier-load scaling regimes.
\item While the above results suggest that one may need to know the system scaling behavior \emph{a priori} to select an asymptotically optimal policy, we show that this is not the case.
Specifically, we define a policy called \thrpolicy that switches between LPF and SERPT depending on how many jobs are in the system.  We show in Theorem \ref{thm:thr:opt} that \thrpolicy can perform well in both \shw and \snds scaling regimes by deferring parallelizable work (using LPF) when there are few jobs in the system and prioritizing short jobs (using SERPT) when there are many jobs in the system.
\item Finally, we examine extensions to our theoretical results via simulation in Section \ref{sec:eval}.
\end{itemize}

\section{Model}\label{sec:model}

We model a system with $k$ identical servers managed by a centralized scheduler.
Arriving jobs are placed in a central queue.
We consider the case where jobs are preemptible and malleable, meaning that the scheduler can change the number of servers on which a job runs over time.
While more complex architectures may exist, this scenario is common in a wide range of real-world systems that process parallelizable jobs.

Without loss of generality (WLOG), we assume each server runs at speed~1.
We assume that each job belongs to one of $\ell$ \emph{job classes}.
Each job class, $i$, has its own job size distribution, denoted by the random variable $S_i$, and parallelizability level, $c_i$.

A job's {\em size} denotes its running time (service time) on a {\em single} server.
We generally assume that job sizes are unknown and exponentially distributed, i.e., 
$$S_i \sim \mbox{Exp}(\mu_i) \quad \forall 1\leq i \leq \ell.$$
Our theoretical results assume exponential job sizes, but we discuss other size distributions in Section \ref{sec:eval:gen}.  

Class $i$ jobs have a  \emph{parallelizability level} of $c_i$, meaning a class $i$ job can run on up to $c_i$ servers.
When a class $i$ job is run on $m\leq c_i$ servers, the job is served at rate $m$.
As a result, this job's remaining service time is distributed as $\mbox{Exp}(m\cdot\mu_i)$.

We assume class $i$ jobs arrive according to a Poisson process with rate $\lambdak_i$, where the total arrival rate is $\lambdak = \sum_{i=1}^\ell\lambdak_i$. 
We define the system load as
$$\rho\ := \sum_{i=1}^\ell \rho_i \quad\mbox{where}\quad
\rho_i := \frac{\lambdak_i}{k\mu_i} \quad \forall \  1 \leq i \leq \ell.$$
The system is stable under any work-conserving scheduling policy if $\rho<1$~\cite{berg2020optimal}.

\subsubsection*{Asymptotic Analysis}
We analyze the system under several scaling regimes.
In the \emph{conventional heavy traffic} scaling regime, the number of servers, $k$,  is held constant, and we take the limit as $\rho \rightarrow 1$.
We will also consider several {\em scaling regimes} where $k \rightarrow \infty$ and the arrival rate, $\lambdak$, scales with $k$, meaning that load, $\rho$, scales with $k$.  
See Appendix \ref{sec:asymptoticnote} for a summary of the standard asymptotic notation used in this paper.

We assume that the proportion of jobs from each class and the job size distributions remain fixed.
That is, we define the constants $p_1, p_2, \ldots, p_\ell$ such that 
$$\lambdak_i = p_i \lambdak \qquad \forall 1 \leq i \leq \ell.$$


We restrict ourselves to scaling regimes where $\lambdak = \Theta(k)$ so that $\rho$ does not go to 0 as we scale $k$.
%
It is convenient to describe the scaling behavior of $\lambdak$ in terms of $\alpha$, the number of ``spare servers'' worth of capacity.
Specifically, we define $\alpha$ to be
$$\alpha := k(1-\rho).$$
Note that while $\alpha$ and $\rho$ depend on $k$, we omit this dependence from the notation for brevity.
We will write expressions in terms of $\lambdak$ when an explicit dependence on $k$ requires emphasis.

The asymptotic behavior of many queueing systems is known to depend on the asymptotic behavior of $\alpha$. For example, consider a simple $M/M/k$ system.
When $\alpha=\Theta(k)$, the scaling regime is known as the \emph{large-system limit} or \emph{mean-field limit}.
The probability of queueing in an $M/M/k$ goes to 0 in the mean-field limit.
When $\alpha = \Theta(\sqrt{k})$, the scaling regime is known as the Halfin-Whitt regime \cite{halfin1981heavy}.
The probability of queueing in an $M/M/k$ converges to a constant in the Halfin-Whitt regime, but the mean queueing time goes to 0.
When $\alpha = \Theta(1)$, the scaling regime is known as the non-degenerate slowdown (NDS) regime \cite{atar2012diffusion,gupta2019load}.
The probability of queueing in an $M/M/k$ goes to 1 in the NDS regime, and the mean queueing time converges to a constant.
We refer to all scaling regimes where $\alpha = \omega(\sqrt{k\log k})$ as \emph{\shw} scaling regimes.
\footnote{The \shw regime was first introduced in \cite{liu2020steady} and was then subsequently discussed in \cite{varma2023power,storm2023diffusion}.  The original definition requires $\alpha=\Theta(k^b)$ for any $.5 < b < 1$.  Our definition is more general in that it requires just $\sqrt{\log k}$ separation from the Halfin-Whitt regime and also includes mean-field scaling regimes.}
In the \shw regime, the probability of queueing in an $M/M/k$ goes to 0.
We refer to all scaling regimes where $\alpha = o(1)$ as \emph{\snds} regimes.
The probability of queueing in an $M/M/k$ goes to 1 and the mean queueing time goes to $\infty$ in \snds regimes.
Note that, under all \snds scaling regimes, $\rho \rightarrow 1$.
Under the \shw scaling regimes, $\rho$ either goes to 1 or converges to a constant less than 1, as in mean-field scaling.
{\em It is not clear which, if any, of the results from the simple $M/M/k$ system hold when scheduling parallelizable jobs.}




While $\rho_i$ depends on $k$ for each class $i$, $\rho_i$ asymptotically approaches a subcritical load.
Specifically, let
$$\rho_i^* := \kinfty \rho_i = \lambdak p_i \E[S_i] \qquad \forall 1 \leq i \leq \ell$$
where $0<\rho_i^* < 1$ for any class $i$, and $\sum_{i=1}^\ell \rho^*_i \leq 1.$
Further, $\sum_{i=1}^\ell \rho^*_i = 1$ in any heavy-traffic scaling regime.

\subsubsection*{Scaling parallelizability levels}
We previously stated that each job class, $i$, has an associated parallelizability level $c_i$.
For some job classes, this parallelizability level results from algorithmic limitations or limitations of a job's software implementation.
For these job classes, the parallelizability level will not change as the system scales.
However, other more-parallelizable classes of jobs can be easily reconfigured as the system scales to leverage newly available servers. 
We therefore allow the value $c_i$ to scale with $k$ for some classes of jobs.
In general, we allow each $c_i$ to be any non-decreasing function of $k$.
In this way, we can represent a variety of real-world system constraints that might cause a job's parallelizability level to grow sublinearly with $k$. 
For example, in a system with a complex network topology, it is common that only a fraction of the servers (e.g., $\sqrt{k}$) can communicate directly with one another.
Hence, for a class of jobs with high communication overhead, the parallelizability level of the jobs might scale as $\sqrt{k}$.

WLOG, let the classes be asymptotically ordered by parallelizability so that, for sufficiently large $k$,
$c_1 \leq c_2 \leq \ldots \leq c_\ell.$
Furthermore, we assume that $c_1=O(1)$, so that at least one class of jobs that requires a constant number of servers.
This mild technical assumption holds in the vast majority of real-world systems.
For simplicity, we assume that $c_1=1$, but setting $c_1$ to some other constant does not affect our results.

\subsubsection*{Scheduling Policies}
A \emph{scheduling policy} must decide, at every moment in time, how many servers to allocate to each job in the system.
Let $\{N(t)\}$ be the stochastic process denoting the number of jobs in the system at time $t\geq0$.
Because we only consider scheduling policies that do not idle servers unnecessarily, $\rho<1$ implies that $\{N(t)\}$ is a positive-recurrent, irreducible continuous-time Markov chain.
By ergodicity, there exists a distribution, $N$, such that $N(t) \overset{a.s.}\rightarrow N$ as $t \to \infty$.
$N$ is the unique stationary distribution of the Markov chain and $\E[N] < \infty$.
We refer to $\E[N]$ as the \emph{mean number of jobs in system}.
We define $\{N_i(t)\}$, $N_i$, and $\E[N_i]$ analogously, but to describe class-$i$ jobs in the system.
This gives $\E[N] = \sum_{i=1}^\ell \E[N_i]$.
Let $X=(N_1, N_2, \ldots, N_\ell)$ denote the full state of the system in stationarity.

Let $T$ denote the \emph{response time} --- the time from when a job arrives until it completes --- of a job that arrives to a system whose state is drawn according to $X$.
We refer to $\E[T]$  as the \emph{mean response time} across jobs.
By Little's law, $\E[T] = \frac{\E[N]}{\lambdak} < \infty$.
We similarly define $\E[T_i]$ to be the mean response time of a class-$i$ job, where $\E[T] = \sum_{i=1}^\ell \E[T_i]p_i$ by conditioning.
Our goal is to find scheduling policies that minimize $\E[T]$ (or, equivalently, minimize $\E[N]$). 
We assume the scheduling policy knows the class information ($\mu_i$ and $c_i$) for each job, as well as the system state. 

Let $\ET{\pi}$ denote the mean response time under any scheduling policy, $\pi$.
$\pi$ is \emph{asymptotically optimal} if
$$\kinfty \frac{\ET{\pi'}}{\ET{\pi}} \geq 1$$ for all policies $\pi'$.
Again, the asymptotic optimality of a policy generally depends on the asymptotic behavior of $\alpha$.
Asymptotic optimality is defined analogously in the conventional heavy traffic limit.

Our results primarily focus on analyzing two policies: one based solely on deferring parallelizable work, and the other based solely on favoring short jobs.
The policy that defers parallelizable work is called Least-Parallelizable-First (LPF).
LPF preemptively prioritizes the jobs with the lowest parallelizability levels.  For example, LPF will allocate as many servers as possible to jobs belonging to the first $i$ classes before allocating any leftover servers to jobs in the $i+1$st class.
The policy that favors short jobs is called the Shortest-Expected-Remaining-Processing-Time (SERPT) policy \cite{scully2020optimal}, which at all times 
 preemptively prioritizes the job with the smallest expected remaining size.  Recall that a job's size is its running time on a single server.  The SERPT policy allocates as many servers as can be used to the job with the smallest expected remaining size before allocating any servers to the next larger size job.

We also introduce a new policy called Threshold (THRESH), which is the first policy to dynamically combine the orthogonal goals of deferring parallelizable work and favoring short jobs.  THRESH runs LPF when the number of jobs in the system is small, and runs SERPT when the number of jobs is large.  
We will fully define THRESH in Section \ref{sec:disc}.

In analyzing these policies, we use a combination of coupling and drift arguments. 

\section{Scheduling Parallelizable Jobs in \shw Regimes}\label{sec:shw}


In this section, we consider scheduling in the \shw scaling regimes described in Section \ref{sec:model}.
We begin by showing that LPF, the policy that defers parallelizable work by prioritizing the least parallelizable jobs, is asymptotically optimal in this case.
Intuitively, in the lighter-load \shw scaling regimes where queueing is rare, the benefit of prioritizing short jobs ahead of long jobs is reduced.
Furthermore, because there will frequently be fewer than $k$ jobs in the system under these scaling regimes, it is important to ``save up" the more parallelizable jobs, which are better at occupying many servers when there are few jobs.
LPF performs well because it keeps more servers occupied than other policies.

Given that there is little queueing in \shw regimes, one might expect a wide variety of policies to be asymptotically optimal.
However, we also show the surprising result that SERPT can be asymptotically suboptimal in \shw regimes.
If SERPT happens to prioritize very parallelizable jobs that fully utilize the system, all the other jobs in the system are forced to queue unnecessarily.
The system is then underutilized when running less-parallelizable jobs.
Hence, the choice of scheduling policy greatly affects the overall mean response time of the system, even in \shw scaling regimes.

We begin by showing the asymptotic optimality of LPF in Theorem \ref{thm:if:opt}.

%
%
%

\begin{restatable}{theorem}{ifsub}\label{thm:if:opt}
If $\alpha = \omega(\sqrt{k\log k})$, then for any scheduling policy $\pi$,
$$\kinfty \frac{\ET{\pi}}{\ET{\LPF}} \geq 1.$$
That is, LPF is asymptotically optimal with respect to mean response time in \shw scaling regimes.
\end{restatable}
We will prove Theorem \ref{thm:if:opt} via Lemmas \ref{lem:shw:inelastic}, \ref{lem:shw:in:opt}, and \ref{if:shw:elastic:opt}.
We first outline our argument at a high level.
Then, we prove the necessary lemmas before proving Theorem \ref{thm:if:opt}.

We will prove that, for each class of jobs $i$, the mean response time of class $i$ jobs under LPF approaches some lower bound as $k \rightarrow \infty$, and thus the overall mean response time under LPF is asymptotically optimal.
When analyzing LPF in this section, we will omit the LPF superscript for the sake of readability.

For any class, $i$, it is desirable to have more than $\rho_i k$ servers working on class $i$ jobs.
In this case, class $i$ jobs are being completed at a rate greater than $\mu_i \rho_i k = \lambdak_i$, and hence the number of class $i$ jobs in the system will tend to decrease.
Given a system with $\alpha$ spare servers, it is theoretically possible to guarantee that $\rho_i k + \frac{\alpha}{\ell}$ servers are available to class $i$ jobs for all classes, $i$.
Our goal is to prove that LPF guarantees this property with high probability, and therefore it is usually possible to complete class $i$ jobs at a very fast rate.
We then argue that this is sufficient to keep the average number of class $i$ jobs in the system very low.

We divide the $\ell$ job classes into two cases.
First, we consider all the job classes whose level of parallelism, $c_i=O(1)$.
We call these \emph{inelastic} (I) jobs, because their level of parallelism does not scale with $k$.
Second, we consider all the job classes whose level of parallelism $c_i=\omega(1)$.
We refer to these as \emph{elastic} (E) jobs.

To handle any inelastic job class, $i$, we argue in Lemma \ref{lem:shw:inelastic} that the probability of having a large number of jobs belonging to the first $i-1$ job classes goes to 0 very quickly as $k \to \infty$.
As a result, LPF almost always devotes a sufficient number to the class $i$ jobs.
This allows us to show in Lemma \ref{lem:shw:in:opt} that the mean response time of class $i$ jobs is asymptotically optimal.

We use a similar argument to prove the asymptotic optimality for elastic jobs in Lemma \ref{if:shw:elastic:opt}.
However, we can prove this result for all elastic classes simultaneously instead of treating the classes one by one as we did for the inelastic classes.
The proof of Theorem \ref{thm:if:opt} follows easily from Lemmas \ref{lem:shw:in:opt} and \ref{if:shw:elastic:opt}.

The key to proving these lemmas is that we have $\frac{\alpha}{\ell}$ spare servers to devote to each class.
Roughly, this means the LPF policy can tolerate fluctuations in the number of class $i$ jobs on the order of $\alpha$ without having these jobs interfere significantly with the service of lower priority jobs.
In \shw scaling regimes, we assume $\alpha=\omega(\sqrt{k\log k})$.
Hence, our proofs use drift arguments to recover a well-known form of result --- that fluctuations in the number of class $i$ jobs larger than $O(\sqrt{k})$ are rare as $k\to\infty$.

We define $\beta_i$ to be the smallest number of class $i$ jobs that can fully utilize $\rho_i k + \frac{\alpha}{\ell}$ servers, giving
$$\beta_i = \Bigl\lceil\frac{\rho_i k + \frac{\alpha}{\ell}}{c_i}\Bigr\rceil.$$
We begin with the following lemma that the number of class $i$ jobs rarely exceeds $\beta_i$ for any inelastic class.

\begin{restatable}{lemma}{shw-inelastic-tail}\label{lem:shw:inelastic}
Let $I$ be the highest number such that $c_I=O(1)$.
For any class of jobs, $i\leq I$, let $\Eps_i$ be the set of states defined as
$$\Eps_i=\{(n_1,n_2, \ldots n_\ell)\colon n_j<\beta_j,\forall 1\leq j \leq i\}.$$
In \shw scaling regimes,
\begin{equation}\Pr((N_1, N_2, \ldots, N_\ell) \notin \Eps_i) = o\left(\frac{1}{k^3}\right) \quad\forall i\leq I.\label{eq:shw}\end{equation}
\end{restatable}
\begin{proof}
We will prove \eqref{eq:shw} by induction on the job class, $i$.
Specifically, we will prove the following two properties by simultaneous induction.
First, we will show there exist constants $\const{d}{i}{j}$ and $\const{f}{i}{j}>0$ such that
\begin{equation}\Pr(N_i \geq \beta_i) = \sum_{j=1}^i \const{f}{i}{j}\left(\frac{k}{\alpha}\right)^{j-1}e^{-\const{d}{i}{j}\alpha^2/k} \quad \forall i\leq I.\label{eq:shw:pt1}\end{equation}
Second, for some constants $\const{g}{i}{j}>0$ and $\const{h}{i}{j}>0$, we will show that
\begin{equation}\Pr((N_1, N_2, \ldots, N_\ell) \notin \Eps_i) \leq \sum_{j=1}^i \const{g}{i}{j}\left(\frac{k}{\alpha}\right)^{j-1}e^{-\const{h}{i}{j}\alpha^2/k} \quad \forall i\leq I.\label{eq:shw:pt2}\end{equation}
Once we have established (2) and (3), we note $\alpha=\omega\left(\sqrt{k\log k}\right)$ in the \shw scaling regimes and thus $\frac{\alpha^2}{k} = \omega\left(\log k\right)$.
Note that, for any constant $x>0$, $\frac{\alpha^2}{k} = \omega\left( x \log k\right) = \omega\left(\log k^x \right)$.
Hence, by setting $x=\frac{j+2}{\const{h}{i}{j}}$, \eqref{eq:shw:pt2} implies
$$\const{g}{i}{j}\frac{k}{\alpha}^{j-1}\cdot e^{-\const{h}{i}{j}\alpha^2/k} = o\left(\frac{1}{k^3}\right)$$
for any $i,j \leq I$, as desired. \footnote{This is where we use the $\sqrt{\log k}$ separation from the Halfin-Whitt regime in our argument about the \shw regime.}
We require \eqref{eq:shw:pt1} as part of our inductive proof.

Our inductive proof of (2) and (3) relies on the drift bounds established in \cite{wang2022heavy}.
These bounds require one to define a Lyapunov function $V(x)$, that takes a system state $x=(n_1,n_2, \ldots n_\ell)$ and returns a real number.
The goal is to choose a Lyapunov function such that, for a state $x$, the instantaneous rate of change of $\E[V(X) \mid X=x]$, known as the \emph{Lyapunov drift},  is almost always negative.
Specifically, if the instantaneous rate of change of $\E[V(X) \mid X=x]$ is negative whenever $V(x) > B$, \cite{wang2022heavy} provides bounds on $\Pr(V(X) > B)$ (see Appendix \ref{sec:driftbound} for a precise statement).
We set $V(x)=n_1$ to bound $\Pr(N_1 > B)$ and prove the base case for our induction ($i=1$).
To prove our claims inductively for class $i+1$ jobs, we set $V(x)=n_{i+1}$.
Here, we must account for the fact that there may be many lower-class jobs in the system preventing the class $i+1$ jobs from running and causing positive drift, even when $V(x)$ is large.
To handle this issue, we exploit a stronger drift bound (see Appendix \ref{sec:app:mod} for a precise statement) that tolerates some positive drift as long as the states with positive drift occur with bounded probability.
The problematic states in this scenario are exactly $\{x \notin \Eps_{i}\}$.
Our inductive hypothesis holds that $\Pr(X\notin \Eps_{i})$ is small, resulting in the desired bounds on $\Pr(N_{i+1} > B)$ and $\Pr(X\notin \Eps_{i+1})$.

To prove both statements when $i=1$, we define the Lyapunov function $V(x)=n_1$.
When $n_1 \geq \rho_1 k + \frac{\alpha}{2\ell}$, the drift of the Lyapunov function, $\Delta V$, is
$$\Delta V(x) = \lambdak_1 - \mu_1 \cdot n_1 \leq \lambdak_1 - \mu_1 (\rho_I k + \frac{\alpha}{2\ell}) = -\mu_1 \frac{\alpha}{2\ell}.$$
Hence, the drift bound from Appendix \ref{sec:driftbound} yields
\begin{align*}
Pr(N_1 \geq \beta_1) \leq Pr\left(N_1 \geq \rho_1 k + \frac{\alpha}{\ell}\right) \leq \const{f}{1}{1}e^{-\const{d}{1}{1} \alpha^2/k}
\end{align*}
for some constant $\const{d}{1}{1}>0$ and $\const{f}{1}{1}=1$.
This proves both \eqref{eq:shw:pt1} and \eqref{eq:shw:pt2} when $i=1$.

Now, we inductively assume that \eqref{eq:shw:pt1} and \eqref{eq:shw:pt2} hold for the first $i$ job classes.
If the Lyapunov function $V(x)=n_{i+1} > \rho_{i+1}k + \frac{\alpha}{2\ell}$, the drift is positive if and only if $x\notin\Eps_{i}$, otherwise the drift is negative.
We can therefore use the modified drift bound from Appendix \ref{sec:app:mod} to get the bound
$$\Pr(N_{i+1} \geq \beta_{i+1}) \leq \const{f}{i+1}{i+1}e^{-\const{d}{i+1}{i+1}\alpha^2/k} + O\left(\frac{k}{\alpha}\right)Pr(X \notin \Eps_i)$$
for some constants $\const{d}{i+1}{i+1}>0$ and $\const{f}{i+1}{i+1}>0$.

Using our inductive hypothesis to bound $\Pr(X \notin \Eps_i)$, we see that there exist some new constants $\const{f}{i+1}{j}$ and $\const{d}{i+1}{j}$ for all $1\leq j \leq i+1$ such that \eqref{eq:shw:pt1} holds for the $i+1$st class of jobs.

Finally, to prove \eqref{eq:shw:pt2} for the $i+1$st class, we note that
$$\Pr(X \notin \Eps_{i+1}) = Pr\left(\bigcup_{j=1}^{i+1} N_j > \beta_j\right) \leq Pr(X \notin \Eps_{i}) + Pr(N_j > \beta_{i+1}) \leq \sum_{j=1}^{i+1} \const{g}{i+1}{j}\left(\frac{k}{\alpha}\right)^{j-1}e^{-\const{h}{i+1}{j}\alpha^2/k}$$
for some new set of constants $\const{g}{i+1}{j}>0$ and $\const{h}{i+1}{j}>0$.
This completes the proof by induction.
\end{proof}
We now use Lemma \ref{lem:shw:inelastic} to prove the optimality of the inelastic response times via a rate conservation argument.

\begin{restatable}{lemma}{shw-inelastic-opt}\label{lem:shw:in:opt}
For any class $i\leq I$, if $\alpha=\omega(\sqrt{k\log k})$, 
then
$$\kinfty \ET[i]{LPF} = \frac{1}{c_i\mu_i}.$$
\end{restatable}
\begin{proof}
Let $K_i$ be the number of servers that process class $i$ jobs under LPF in a random state $X=(N_1, N_2, \ldots, N_i)$ in the stationary regime. 
Let $\theta_i =  \lceil \frac{\rho_i k + \log k}{c_i} \rceil$ be the number of class $i$ jobs required to use slightly more than the average number of class $i$ servers.  The number of servers used by $\theta_i$ class $i$ jobs is at most $\rho_i k + \log k + c_i$ due to rounding.
We also define $\delta_i=(N_i-\theta_i)^+$, where $(\cdot)^+$ denotes $\max(\cdot,0)$ (i.e. the \emph{positive part}).

Intuitively, while class $i$ jobs use $\rho_i k$ servers on average, Lemma \ref{lem:shw:inelastic} suggests that it should be rare for the class $i$ jobs to use $\rho_i k + \log k$ servers.
We will use Lemma \ref{lem:shw:inelastic} to show that $\E[\delta_i] \rightarrow 0$ and hence $\E[N_i]$ is at most $\theta_i$ as $k$ becomes large.  
We then apply Little's law to this bound to finish our proof.

Formally, we will make a rate conservation argument using the Lyapunov function $V(x)=((n_i - \theta_i)^+)^2$ so that $V(X)=\delta_i^2$.
Note that  $\delta_i$ does not change for small values of $N_i$.  If an arrival occurs when there are fewer than $\theta_i$ jobs, nothing changes.  Similarly, unless there are strictly more than $\theta_i$ jobs, nothing changes on a departure either.  This gives:
\begin{align*}
\1{N_i \geq \theta_i} \lambdak_i (2\delta_i +1) + \1{N_i > \theta_i}\mu_i K_i ( -2\delta_i +1) &= 0\\
\1{N_i \geq \theta_i} \lambdak_i 2\delta_i + \1{N_i \geq \theta_i} \lambdak_i + \1{N_i > \theta_i}\mu_i K_i &= \1{N_i > \theta_i}\mu_i K_i 2\delta_i\\
2 \1{N_i \geq \theta_i} \lambdak_i \delta_i + \1{N_i \geq \theta_i} \lambdak_i + \1{N_i > \theta_i}\mu_i K_i &= 2\1{N_i \geq \theta_i}\mu_i K_i \delta_i
\end{align*}
Taking expectations then gives
\begin{align}
\E[2\1{N_i \geq \theta_i} \lambdak_i \delta_i + \1{N_i \geq \theta_i} \lambdak_i + \1{N_i > \theta_i}\mu_i K_i] = \E[2\1{N_i \geq \theta_i}\mu_i K_i \delta_i ]\\
\E[\1{N_i \geq \theta_i} \lambdak_i \delta_i + \1{N_i \geq \theta_i} \lambdak_i] = \E[\1{N_i \geq \theta_i}\mu_i K_i \delta_i ]\label{eq:arrivals}
\end{align}
where \eqref{eq:arrivals} used the fact that $\E[\1{N_i > \theta_i}\mu_i K_i] = \E[\1{N_i \geq \theta_i}\lambdak_i]$.  This is easily shown using a rate conservation argument with $V(x)=\delta_i$, which yields that $\1{N_i \geq \theta_i}\lambdak_i - \1{N_i > \theta_i}\mu_i K_i = 0$ in stationarity.
Now, noting that $\1{N_i \geq \theta_i}\delta_i = \delta_i$, we have
\begin{align}
\lambdak_i Pr(N_i \geq \theta_i)  + \lambdak_i \E[\delta_i]  = \mu_i \E[K_i \delta_i].\label{eq:rate}
\end{align}

If $K_i$ and $\delta_i$ were independent, $\mu_i\E[K_i\delta_i]$ would equal $\lambdak_i \E[\delta_i]$ since we know the long run completion rate of type $i$ jobs equals the long run arrival rate of type $i$ jobs in a stable system.
Intuitively, however, $\delta_i$ and $K_i$ should be positively correlated --- we should expect to use more servers on type $i$ jobs when more type $i$ jobs are in the system.
Hence, to turn \eqref{eq:rate} into an upper bound on $\E[\delta_i]$, we will lower bound $\E[K_i\delta_i]$ to be significantly larger than  $\lambdak_i \E[\delta_i]$.
$K_i$ and $\delta_i$ only move inversely when many lower-class jobs enter the system and prevent class $i$ jobs from running.
Hence, we can use our upper bound on $\Pr(X\notin \Eps_{i-1})$ to argue that $K_i$ and $\delta_i$ are positively correlated and thus
\begin{equation}\E[K_i\delta_i] \geq (\E[K_i] + \log k)\E[\delta_i] - O(1) = (\rho_i k + \log k)\E[\delta_i] - O(1), \quad \forall i\leq I.\label{eq:product}\end{equation}
This is shown formally in Appendix \ref{sec:momentbound}.
Plugging \eqref{eq:product} into \eqref{eq:rate} yields
\begin{align}
\lambdak_i + O(1) &\geq (\mu_i(k\rho_i + \log k))  - \lambda_i)\E[\delta_i]= \mu_i \log k\E[\delta_i]\\
\frac{\lambdak_i + O(1)}{\mu_i \log k} &\geq \E[\delta_i]
\end{align}
for sufficiently large $k$.
Hence,
\begin{align}
\E[N_i] &\leq \E[\delta_i] + \theta_i \leq o(k) + \theta_i\\
\E[N_i] &\leq o(k) + \frac{\rho_i k }{c_i}.
\end{align}
Applying Little's law gives
$$\E[T_i] = o(1) +\frac{1}{c_i \mu_i}.$$ 
Hence, for the first $I$ classes, the mean response time of each class converges to a service time plus an additive term that goes to 0 as $k \rightarrow \infty$.
Note that this mean service time represents the lowest possible mean service time for a class $i$ job, and is thus a lower bound on the mean response time for class $i$ jobs.
\end{proof}

At this point, we have shown that the I jobs have asymptotically optimal mean response time in \shw regimes.
Now, we handle the remaining classes whose level of parallelizability scales with $k$.
\begin{restatable}{lemma}{if-shw-elastic-opt}\label{if:shw:elastic:opt}
For any class $i>I$, if $\alpha=\omega(\sqrt{k\log k})$, then
$$\kinfty \ET[i]{LPF} = 0.$$
\end{restatable}
\begin{proof}
Recall that $c_i= \omega(1)$ for all classes of elastic (E)  jobs.
We will bound the response time of all elastic classes simultaneously.
Let
$$\rho_E = \sum_{j=I+1}^{\ell} \rho_j \mand \lambdak_E = \sum_{j=I+1}^{\ell} \lambdak_j.$$
Let $K_E$ be the combined number of servers used on any elastic job, and let $\mu_{min} = \min_{i\in[I+1,\ell]} \mu_i$.
We will bound $\E[N_E]$, the expected number of total jobs of class $I+1$ or higher.
Following our earlier argument, we want to exploit our bound on $\Eps_{I}$ to get an upper bound on $\E[N_E]$.

Rather than tracking the number of jobs directly, we will track the expected sum of the sizes of elastic jobs in the system.
We refer to this quantity as the expected \emph{work} from elastic jobs, $W_E$.
This quantity can be written as 
$$W_E = \sum_{j=I+1}^{\ell} \frac{N_j}{\mu_j}.$$
Let $\theta_E = \frac{\rho_E k + \log k}{c_{I+1} \mu_{min}}$ and let $\delta_E = (W_E - \theta_E)^+$.
Furthermore, let $\delta_E(x)$ be the value of $\delta_E$ given that the system is in state $x$.

Note that we divide by $c_{I+1}$ to be pessimistic about how many servers each job uses.
We will develop a bound on $\E[\delta_E]$ and use this bound to recover a bound on $\E[N_E]$.
The key to this lemma is that $\theta_E=o(k)$ because $c_{I+1}=\omega(1)$.
Hence, we can stand to pick up a $\theta_E$ term in our bound on $\E[N_E]$, and this term will be dominated by $\lambdak_E$ when applying Little's law.

We will now use a rate conservation argument with the Lyapunov function $V(x)=\left(\delta_E(x)\right)^2$.  This yields
$$\1{W_E\geq\theta_E}\left(\sum_{j={I+1}}^{\ell} \frac{2\lambdak_j}{\mu_j}\delta_E + \frac{\lambdak_j}{\mu^2_{j}}\right) + \1{W_E>\theta_E}\left(\sum_{j={I+1}}^{\ell} -2k_j\delta_E + \frac{k_j}{\mu_j}\right) = 0.$$
Taking expectations, we get
\begin{align}
\E[\1{W_E\geq\theta_E}\delta_E k\rho_E] + O(k) &= \E[\1{W_E\geq\theta_E}\delta_EK_E]\\
k\rho_E\E[\delta_E] + O(k) &= \E[\delta_EK_E]\label{eq:econs}
\end{align}
where the indicators on both sides vanish because they are redundant, and the terms without $\delta_E$ are aggregated into the $O(k)$.
Hence, lower bounding $\E[\delta_EK_E]$ will give an upper bound on $\E[\delta_E]$.

Our argument now follows the proof of \eqref{eq:product} (see Appendix \ref{sec:momentbound}) to show that, for sufficiently large $k$,
$$\E[\delta_EK_E] \geq (\rho_E k + \log k) \E[\delta_E] - O(1).$$
Adapting Lemma \ref{lem:lb} to prove this claim is straightforward and makes use of our bound on $\Pr(X\notin\Eps_I)$.



Plugging this lower bound into \eqref{eq:econs} gives
\begin{align}
O(k) + O(1) &\geq (\rho_E k + \log k)\E[\delta_E] - \rho_E k \E[\delta_E]\\
\frac{O(k) + O(1)}{\log k} &\geq \E[\delta_E]\\
o(k) &= \E[\delta_E].
\end{align}

To conclude, we note that 
$$\E[W_E] \leq \E[\delta_E] + \E[\theta_E] = \E[\delta_E] + o(k).$$
Let $\mu_{max} = \max_{i\in[I+1,\ell]}\mu_i$. 
We then have
$$\E[N_E] = \sum_{j=I+1}^{\ell} \E[N_j] \leq \sum_{j=I+1}^{\ell} \frac{\E[N_j]\mu_{max}}{\mu_j} = \mu_{max}\E[W_E].$$
Hence,
$$\E[N_E] \leq \mu_{max} (\E[\delta_E] + o(k)) = o(k).$$
Applying Little's Law, we have
$$\kinfty \E[T_E] = \kinfty \frac{\E[N_E]}{\lambdak_E} = 0.$$
This implies our claim, that the mean response time of class $i$ jobs goes to 0 for all $i>I$.
\end{proof}

We now combine the above lemmas to complete our proof of Theorem \ref{thm:if:opt}.
\begin{proof}[Proof of Theorem \ref{thm:if:opt}]
By Lemmas \ref{lem:shw:in:opt} and \ref{if:shw:elastic:opt}, for any class, $i$, $\ET[i]{LPF}$ converges to a lower bound as $k$ becomes large, since the mean response time of a job is lower bounded by its mean service time.  Hence, the overall mean response time across jobs, which is a linear combination of the mean response times for each class, also converges to a lower bound.  For any policy, $\pi$, this gives
$$\kinfty\frac{\ET{\pi}}{\ET{LPF}}\geq 1,$$
and thus LPF is optimal in \shw scaling regimes.
\end{proof}

We conclude this section with Theorem \ref{thm:srptsub}, which shows that SERPT is asymptotically suboptimal in \shw scaling regimes.
That is, LPF outperforms SERPT (sometimes significantly) in \shw scaling regimes.

\begin{restatable}{theorem}{srptsub}\label{thm:srptsub}
There exists a set of job classes such that, when $\alpha = \omega(\sqrt{k\log{k}})$,
$$\kinfty \frac{\ET{\SERPT}}{\ET{\LPF}} > 1.$$
That is, SERPT can be strictly suboptimal in the \shw regime.
\end{restatable}
\begin{proof}
Consider the case where $\ell=2$ with $c_1=1$, $c_2=k$, and $\mu_2 > \mu_1$.
We can easily see that the mean response time of class 2 jobs goes to 0 under SERPT in the \shw regime.
Hence, to prove the suboptimality of SERPT, we must show that $\kinfty \E\left[T_1^{\LPF}\right] < \kinfty \E\left[T_1^{\SERPT}\right]$.
Lemma \ref{lem:shw:in:opt} shows that, under LPF in the \bhw regime, the mean response time of a class 1 job converges to $\frac{1}{\mu_1}$.
We will show that SERPT does not attain this lower bound on the mean response time of class 1 jobs in the \bhw regime.

We consider the case where $\mu_2 > \mu_1$, so SERPT gives strict priority to class 2 jobs.
Consider the behavior of one particular class 1 job.
This class 1 job will be served for some random amount of time, $S_1$.
The queueing time of the job is therefore \emph{at least} the amount of time spent waiting behind class 2 jobs that arrive during its $S_1$ time in service.
We define $A^{(2)}_{S_1}$ to be the number of class 2 arrivals during a (random) class 1 service time.
We can then write
$$\E\left[T_1^{\SERPT}\right] \geq \E\left[\sum_{i=1}^{A^{(2)}_{S_1}} \frac{S_{2_i}}{k}\right] + \E[S_1].$$
Clearly, each $S_{2_i}$ is i.i.d, and also independent of $A^{(2)}_{S_1}$, so we have
$$\E\left[T_1^{\SERPT}\right] \geq \E\left[A^{(2)}_{S_1}\right]\E\left[\frac{S_2}{k}\right] + \E[S_1].$$
We can then see that 
\begin{align*}
\E\left[A^{(2)}_{S_1}\right] &= \int_0^\infty \mu_1 e^{-\mu_1 t}\E\left[A^{(2)}_{S_1} \mid S_1=t\right]dt = \int_0^\infty  t\cdot\mu_1 e^{-\mu_1 t}\cdot\lambdak_2 dt=\lambdak_2\E[S_1].
\end{align*}
Thus,
$$\E\left[T_1^{\SERPT}\right] \geq \frac{\lambdak_2}{k}\E[S_2]\E[S_1] + \E[S_1] = \left(\rho_2 + 1\right)\E[S_1]$$
and 
$$\kinfty \E\left[T_1^{\SERPT}\right] \geq \frac{(1+\rho_2^*)}{\mu_1} > \kinfty\E\left[T_1^{\LPF}\right].$$

Because SERPT achieves suboptimal mean response time for class 1 jobs as $k \rightarrow \infty$, the overall mean response time under SERPT is suboptimal in \shw scaling regimes.
\end{proof}


\section{Scheduling Parallelizable Jobs in \snds Regimes}\label{sec:snds}

LPF performs optimally in \shw scaling regimes because it defers parallelizable work in order to keep servers occupied.
SERPT, on the other hand, performs suboptimally in these regimes.
In \snds scaling regimes, this trend is reversed.
Namely, in \snds scaling regimes, the mean queueing time becomes the dominant term in the mean response time.
Because SERPT minimizes mean queueing time by prioritizing short jobs ahead of long jobs, we can show it is asymptotically optimal in this case.
Not only does LPF fail to minimize mean queueing time, but the benefits of deferring parallelizable work are diminished in \snds regimes.
Here, there are usually enough jobs in the system to keep all servers occupied, even without deferring parallelizable work.
Hence, we show that LPF is suboptimal in \snds scaling regimes.
Additionally, we will show that both results also apply in the conventional heavy traffic limit, where $k$ is fixed and $\rho \to 1$.

The point of this section will be to prove Theorems~\ref{thm:conventional},\ref{thm:snds}, and \ref{thm:ifsubopt}, which state that SERPT is optimal in conventional heavy traffic and \snds scaling regimes, while LPF can be suboptimal in these cases.  We start by proving a preliminary lemma that will be used repeatedly in our proofs.  

This preliminary lemma, Lemma \ref{lem:conventional:number}, relates the mean number of jobs in a SERPT system to a lower bound on the optimal mean number of jobs in system.
Specifically, we consider the lower bound of an optimal policy for a {\em single-server} preemptive system that runs at speed $k$.
The optimal policy for this single-server system serves jobs one at a time in SERPT order \cite{van1995dynamic}.
We will refer to the optimal policy for the single-server system as SERPT-1, to emphasize that this is SERPT applied to a system with one fast server.
The SERPT policy (without the 1) will continue to denote SERPT in our $k$-server system.

To describe the priority ordering of classes under SERPT and SERPT-1, let $\gamma=(\gamma_1, \gamma_2, \ldots, \gamma_\ell)$ be a sequence such that
$$\mu_{\gamma_i} \geq \mu_{\gamma_{i+1}} \quad \forall 1\leq i < \ell.$$

We now compare the expected number of jobs under SERPT in our $k$-server system to the expected number of jobs in the SERPT-1 system.



\begin{restatable}{lemma}{conventional-number}\label{lem:conventional:number}
$$\EN{\SERPTONE} \leq \EN{\SERPT} \leq \EN{\SERPTONE} + k\cdot \ell.$$
\end{restatable}
\begin{proof}
The first part of our claim, that $\EN{\SERPTONE} \leq \EN{\SERPT}$, holds because
SERPT-1 is optimal for the single-server system.  Any policy from the $k$-server system can be mimicked by the single-server system, so 
SERPT-1 must provide a lower bound on the optimal policy for the $k$-server system.

The rest of this proof is devoted to proving the claim that $\EN{\SERPT} \leq \EN{\SERPTONE} + k\cdot\ell$.
To prove this, we will track the total work in the system, $W$, for each policy.
Specifically, let $W_i^{\pi}(t)$ be the total remaining size of all class $i$ jobs in a system using the scheduling policy $\pi$ at time $t$.
Furthermore, let $W_{\leq \gamma_i}^{\pi}(t)$ be the total remaining size of job \emph{in classes $\gamma_1$ through $\gamma_i$} at time $t$.
Finally, let $\Delta_{\leq \gamma_i}(t)$ be defined as
$$\Delta_{\leq \gamma_i}(t) := W_{\leq \gamma_{i}}^{\SERPT}(t) - W_{\leq \gamma_{i}}^{\SERPTONE}(t).$$

We can take expectations to get 
$$\EW[\leq \gamma_i]{\SERPT}= \EW[\leq \gamma_i]{\SERPTONE}+\E[\Delta_{\leq \gamma_i}].$$
Our goal is to get upper and lower bounds on $\E[\Delta_{\leq \gamma_i}]$, and then use these bounds on the work in system to get a bound on the expected number of jobs in system.

First, it is easy to see that $\E[\Delta_{\leq \gamma_i}]\geq 0$.
Consider any fixed arrival sequence of jobs and their associated sizes.
This sequence corresponds to one possible sample path for all the random variables in the system.
Assume for contradiction that, at some time $t$,
$$W_{\leq \gamma_{i}}^{\SERPT}(t) < W_{\leq \gamma_{i}}^{\SERPTONE}(t).$$
This implies that just before time $t$, the SERPT system completes work on classes $\gamma_1$ through $\gamma_i$ at a faster rate than the SERPT-1 system.
However, the SERPT-1 system always completes this work at the maximal rate unless $W_{\leq \gamma_{i}}^{\SERPTONE}(t)=0$.
This implies that
$$W_{\leq \gamma_{i}}^{\SERPT}(t) < 0,$$
a contradiction.
Because $W_{\leq \gamma_{i}}^{\SERPT}(t) \geq  W_{\leq \gamma_{i}}^{\SERPTONE}(t)$ at any time $t$, for any arrival sequence of jobs, we have
$$\EW[\leq \gamma_i]{\SERPT} \geq \EW[\leq \gamma_i]{\SERPTONE} \qquad \forall 1 \leq i \leq \ell.$$

To derive an upper bound, we again consider any fixed arrival sequence of jobs.
For any value of $i$, we partition this arrival sequence into ``non-busy'' intervals where the SERPT system has fewer than $k$ jobs in classes $\gamma_1$ through $\gamma_i$ and ``busy'' intervals where SERPT has more than $k$ jobs from these classes.
For any time $t_n$ in a non-busy interval, it is clear that $\Delta_{\leq \gamma_i}(t_n)$ is at most the total remaining size of the at most $k$ jobs in the SERPT system. 

We now take expectations over all possible sample paths, giving
$$\E[\Delta_{\leq \gamma_i}(t_n)] \leq \E\left[\sum_{j=1}^{ N^{\SERPT}_{\leq \gamma_i}(t_n)}S^{(j)}_{\gamma_i}\right] \leq k\E[S_{\gamma_i}]$$
where $\E[S_{\gamma_i}]$ is the expected size of a job from the class with the $i$th smallest expected job size.
We are exploiting memorylessness here --- the remaining size of each of the jobs in the SERPT system is exponentially distributed.
Furthermore, because the classes are ordered by their expected job sizes, the remaining size of each job is stochastically less than or equal to $S_{\gamma_i}$.

We now use a similar argument to handle any time $t_b$ in a busy interval.
For any time $t_b$, there exists some $t^*\leq t_b$ that marks the start of the busy interval.
At time $t^*$, $\Delta_{\leq \gamma_i}(t^*)$ is at most the total remaining size of the (exactly) $k$ jobs in the SERPT system.
In the interval $[t^*,t_b]$, the SERPT completes work at rate $k$, the maximal rate of work for both systems, so $\Delta_{\leq \gamma_i}(t_b)\leq \Delta_{\leq \gamma_i}(t^*)$.
We again take expectations to find that, for any time $t_b$, 
$$\E[\Delta_{\leq \gamma_i}(t_b)] \leq  k\E[S_{\gamma_i}].$$

Because our upper bound on $\Delta_{\leq \gamma_i}$ holds for any time $t\geq0$, we have
$$\E[\Delta_{\leq \gamma_i}] \leq  k\E[S_{\gamma_i}] \qquad \forall 1\leq i \leq \ell$$
and therefore
$$\EW[\leq \gamma_i]{\SERPT} \leq \EW[\leq \gamma_i]{\SERPTONE}+k\E[S_{\gamma_i}] \qquad \forall 1\leq i \leq \ell.$$

For class $\gamma_1$ jobs, we now immediately use our upper bound to get
$$\EW[\gamma_1]{\SERPT} \leq  \EW[\gamma_1]{\SERPTONE} + \frac{k}{\mu_{\gamma_1}}.$$
Because job sizes are exponential, we can divide this expression by the mean size of a class $\gamma_1$ job to get
$$\EN[\gamma_1]{\SERPT} \leq  \EN[\gamma_1]{\SERPTONE} + k.$$

For class $\gamma_i$ jobs where $1 < i \leq \ell$, we have
$$\EW[\leq \gamma_i]{\SERPT} \leq  \EW[\leq \gamma_i]{\SERPTONE} + \frac{k}{\mu_{\gamma_i}}\mand \EW[\leq \gamma_{i-1}]{\SERPT} \geq  \EW[\leq \gamma_{i-1}]{\SERPTONE}.$$
We subtract the latter inequality from the former to yield
\begin{align}
\EW[\leq \gamma_i]{\SERPT} - \EW[\leq \gamma_{i-1}]{\SERPT} &\leq \EW[\leq \gamma_i]{\SERPTONE} + \frac{k}{\mu_{\gamma_i}} - \EW[\leq \gamma_{i-1}]{\SERPTONE}\\
\EW[\gamma_i]{\SERPT} &\leq  \EW[\gamma_i]{\SERPTONE} + \frac{k}{\mu_{\gamma_i}}.
\end{align}

We again divide by the average size of a class $\gamma_i$ job to get
$$\EN[\gamma_i]{\SERPT} \leq  \EN[\gamma_i]{\SERPTONE} + k \quad \forall 1 \leq i \leq \ell \implies \EN{\SERPT} \leq \EN{\SERPTONE} + k\cdot \ell$$
as desired.
\end{proof}

We are now ready to use Lemma \ref{lem:conventional:number} to prove Theorems \ref{thm:conventional} and \ref{thm:snds}.
\begin{restatable}{theorem}{conventional}\label{thm:conventional}
Consider a system in the conventional heavy traffic limit as $\rho \rightarrow 1$.
For any policy, $\pi$,
$$\lim_{\rho \rightarrow 1} \frac{\ET{\pi}}{\ET{\SERPT}}\geq1.$$
That is, SERPT is asymptotically optimal in the conventional heavy traffic regime.
\end{restatable}
\begin{proof}
To establish the optimality of SERPT in conventional heavy traffic, we note that for any policy $\pi$, $\E[T^\pi] \geq \ET{\SERPTONE}$.
If we can show that
$$\lim_{\rho \rightarrow 1} \frac{\E[T^{\SERPT}]}{\ET{\SERPTONE}} = 1,$$
then
$$ \lim_{\rho \rightarrow 1} \frac{\E[T^{\pi}]}{\E[T^{\SERPT}]} 
= \lim_{\rho \rightarrow 1} \frac{\E[T^{\pi}]}{\E[T^{\SERPT}]} \cdot \lim_{\rho \rightarrow 1} \frac{\E[T^{\SERPT}]}{\ET{\SERPTONE}} =
\lim_{\rho \rightarrow 1} \frac{\E[T^{\pi}]}{\ET{\SERPTONE}} \geq 1.$$
Hence, we now show that the $\ET{\SERPT}$ converges to the $\ET{\SERPTONE}$ as $\rho \to 1$.

We apply Little's law to the result of Lemma \ref{lem:conventional:number}, which yields
\begin{equation}\ET{\SERPTONE} \leq \ET{\SERPT}  \leq \ET{\SERPTONE} + O(1).\label{eq:etbound}\end{equation}


The mean response time of the SERPT-1 system in conventional heavy traffic can be derived from standard results on priority queueing (e.g., (32.1) in \cite{harchol2013performance}), yielding
$$\rhoone\ET{\SERPTONE}=\Theta\left(\frac{1}{1-\rho}\right).$$
Combining these two results gives
$$\rhoone\frac{\E[T^{\SERPT}]}{\E[T^{\SERPTONE}]} \leq \rhoone\frac{\E[T^{\SERPTONE}] + O(1)}{\E[T^{\SERPTONE}]} = \rhoone1 + O\left(1-\rho\right) = 1.$$
Because the mean response time under SERPT converges to that of the SERPT-1 system, SERPT is optimal in conventional heavy traffic.
\end{proof}

We now extend our proof from conventional heavy traffic to cover all \snds scaling regimes.

\begin{restatable}{theorem}{thm-snds}\label{thm:snds}
If $\alpha=o(1)$, then for any policy, $\pi$,
$$\kinfty \frac{\ET{\pi}}{\ET{\SERPT}}\geq1.$$
That is, SERPT is asymptotically optimal in \snds scaling regimes.
\end{restatable}
\begin{proof}
This argument closely follows the proof of Theorem \ref{thm:conventional}.
We again apply Lemma \ref{lem:conventional:number} to get \eqref{eq:etbound}.
In \snds scaling regimes, we can again use a standard analysis of priority queueing systems to find that $\ET{\SERPTONE}=\Theta\left(\frac{1}{\alpha}\right)$.
This gives the analogous result that
$$\kinfty\frac{\E[T^{\SERPT}]}{\E[T^{\SERPTONE}]} \leq \kinfty\frac{\E[T^{\SERPTONE}] + O(1)}{\E[T^{\SERPTONE}]} = \kinfty 1 + O\left(\alpha\right) = 1.$$

Hence, SERPT is asymptotically optimal in \snds scaling regimes.
\end{proof}

To conclude this section, we show that LPF can fail to minimize mean queueing time, and can therefore be asymptotically worse than SERPT in both conventional heavy traffic and \snds scaling regimes.

\begin{restatable}{theorem}{ifsubopt}\label{thm:ifsubopt}
Consider a system with a constant number of servers, $k_0$, in conventional heavy traffic.
There exists a set of job classes such that 
$$\lim_{\rho \rightarrow 1} \frac{\ET{\LPF}}{\ET{\SERPT}} > 1.$$
Furthermore, in \snds scaling regimes where $\alpha=o(1)$,  there exists a set of job classes such that
$$\lim_{k \rightarrow \infty} \frac{\ET{\LPF}}{\ET{\SERPT}} > 1.$$
That is, LPF can be strictly suboptimal in conventional heavy traffic and \snds scaling regimes.
\end{restatable}
\begin{proof}
Consider the case where $\ell=2$ with $c_1=1$, $c_2=k$, and $\mu_2 > \mu_1$.
We established that SERPT is optimal in conventional heavy traffic by showing that
$$\lim_{\rho \rightarrow 1} \frac{\ET{\SERPT}}{\ET{\SERPTONE}}=1.$$
Given this result, we can prove the first part of our claim by showing that
$$\lim_{\rho \rightarrow 1} \frac{\ET{\LPF}}{\ET{\SERPTONE}}>1 \implies \lim_{\rho \rightarrow 1} \frac{\ET{\LPF}}{\ET{\SERPT}} = \lim_{\rho \rightarrow 1} \frac{\ET{\LPF}}{\ET{\SERPTONE}} \cdot \frac{\ET{\SERPTONE}}{\ET{\SERPT}} =  \lim_{\rho \rightarrow 1} \frac{\ET{\LPF}}{\ET{\SERPTONE}} > 1.$$
We will also make an analogous argument for the case of \snds scaling.
We begin by proving a lower bound on the performance of LPF.

We refer to a single-server system of speed $k$ that serves jobs in LPF order as an \LPFONE system.  
We compare the performance of \LPFONE to the performance of LPF in our original k-server system to get
$\lim_{\rho \rightarrow 1} \frac{\ET{\LPF}}{\ET{\LPFONEM}} \geq 1.$
Similarly, in a \snds scaling regime,
$\lim_{k \rightarrow \infty} \frac{\ET{\LPF}}{\ET{\LPFONEM}} \geq 1.$
These claims follow from a coupling argument similar to the proof of Lemma \ref{lem:conventional:number}.
The main difference is that we are assuming a different priority ordering of the job classes in this case, but this does not affect our coupling argument.

To complete the first part of our claim, we will show that in conventional heavy traffic,
\begin{equation}\lim_{\rho \rightarrow 1} \frac{\ET{\LPFONEM}}{\ET{\SERPTONE}} > 1.\label{eq:ifcmu}\end{equation}
This will imply that
\begin{equation}\lim_{\rho \rightarrow 1} \frac{\ET{\LPF}}{\ET{\SERPTONE}}=\lim_{\rho \rightarrow 1}\frac{\ET{\LPFONEM}}{\ET{\SERPTONE}}\cdot \frac{\ET{\LPF}}{\ET{\LPFONEM}} > 1.\label{eq:ifsingle}\end{equation}

To show \eqref{eq:ifcmu}, we rewrite this claim in terms of class 1 and class 2 response times to get
\begin{align*}
\frac{\ET{\LPFONEM}}{\ET{\SERPTONE}} &= \frac{p_1 \ET[1]{\LPFONEM} + p_2 \ET[2]{\LPFONEM}}{p_1 \ET[1]{\SERPTONE} + p_2 \ET[2]{\SERPTONE}}\\
&=\frac{p_1 \frac{\E[S_1]}{k(1-\rho_1)} + p_2 \ET[2]{\LPFONEM}}{p_2 \frac{\E[S_2]}{k(1-\rho_2)} + p_1 \ET[1]{\SERPTONE}}=\frac{O(1) + p_2 \ET[2]{\LPFONEM}}{O(1) + p_1 \ET[1]{\SERPTONE}}=\frac{\frac{O(1)}{p_1 \ET[1]{\SERPTONE}} + \frac{p_2 \ET[2]{\LPFONEM}}{p_1 \ET[1]{\SERPTONE}}}{\frac{O(1)}{p_1 \ET[1]{\SERPTONE}} + 1}.
\end{align*}

Hence, it suffices to show that 
\begin{equation}\ET[1]{\SERPTONE} = \omega(1)\label{eq:cond1}\end{equation}
and 
\begin{equation}\lim_{\rho \rightarrow 1}\frac{p_2 \ET[2]{\LPFONEM}}{p_1 \ET[1]{\SERPTONE}} > 1.\label{eq:cond2}\end{equation}

From standard techniques for analyzing single-server priority queues, we know that $\ET[1]{\SERPTONE} = \Theta(\frac{\rho}{1-\rho}) = \omega(1)$ as $\rho \rightarrow 1$, and that \eqref{eq:cond2} holds when $\mu_2 > \mu_1$.
Combined, these two conditions prove \eqref{eq:ifcmu} and \eqref{eq:ifsingle}, completing the first part of the theorem.

To prove our second claim, we make an analogous argument about \snds scaling, where all limits are taken as $k \rightarrow \infty$.
Specifically,  $\ET[1]{\SERPTONE} = \Theta\left(\frac{1}{\alpha}\right)$ under \snds scaling, so \eqref{eq:cond1} still holds in this case.
Similarly, we know from standard single-server queueing analysis that
\begin{equation}\kinfty\frac{p_2 \ET[2]{\LPFONEM}}{p_1 \ET[1]{\SERPTONE}} > 1.\label{eq:cond3}\end{equation}
Thus, we can conclude that
$$\lim_{k\rightarrow \infty}\frac{\ET{\LPF}}{\ET{\SERPTONE}} = \lim_{k\rightarrow \infty}\frac{\ET{\LPFONEM}}{\ET{\SERPTONE}}\cdot\frac{\ET{\LPF}}{\ET{\LPFONEM}} = \lim_{k\rightarrow \infty}\frac{\ET{\LPFONEM}}{\ET{\SERPTONE}} > 1.$$

\end{proof}

\section{The \thrpolicy Policy}\label{sec:disc}
%
%
One limitation of our results thus far is that one may not know their system's scaling regime.
Theorems \ref{thm:if:opt} and \ref{thm:snds} suggest that one must know the system scaling behavior in order to select an asymptotically optimal policy.
We therefore propose a \emph{threshold-based policy}, \thrpolicy, that dynamically chooses between LPF and SERPT without knowing the system scaling regime.
\thrpolicy toggles between LPF and SERPT depending on whether the number of jobs in the system exceeds some threshold, $\tau \approx k$. 
When the number of jobs in the system is less than $\tau$, \thrpolicy uses LPF to defer parallelizable work and avoid idling servers.
When the number of jobs in the system is at least $\tau$, \thrpolicy uses SERPT to minimize queueing time.

Given two job classes with $c_1=1$ and $c_2=k$, we show that \thrpolicy is asymptotically optimal in both \shw and \snds regimes.
While our theoretical results are limited to this two-class case, we show in Section \ref{sec:eval} that \thrpolicy also performs well with additional inelastic and elastic classes of jobs.


The exact threshold value, $\tau$, affects the traffic regimes in which \thrpolicy is asymptotically optimal.
We show that when $\tau=\omega(k)$, \thrpolicy is asymptotically optimal in \shw and conventional heavy traffic regimes.
Furthermore, \thrpolicy is asymptotically optimal in any super-NDS regime where $\alpha = o(k/\tau)$.
That is, the threshold value $\tau$ can be any function that is $\omega(k)$, but its exact order affects the optimality result in the super-NDS regime.  For example, we can choose $\tau$ to be $\tau=k\log^* k$, where $\log^*$ is the iterated logarithm function.
Then in the super-NDS regime, \thrpolicy is asymptotically optimal when $\alpha=o(1/\log^* k)$.
Hence, \thrpolicy can be asymptotically optimal in an arbitrarily large portion of the \snds regimes, depending on the choice of $\tau$.
We state this formally in Theorem \ref{thm:thr:opt}, below.
\begin{samepage}
\begin{restatable}{theorem}{thropt}\label{thm:thr:opt}
Consider the problem of scheduling $\ell=2$ classes of parallelizable jobs with $c_1=1$ and $c_2=k$.
Consider the \thrpolicy policy with a threshold value $\tau=\omega(k)$.
For any policy, $\pi$:
\begin{enumerate}
\item If $\alpha = \omega(\sqrt{k\log k})$, then $\kinfty \frac{\ET{\pi}}{\ET{\THRESH}} \geq 1.$

\item In a conventional heavy-traffic scaling regime, 
$\lim_{\rho \rightarrow 1} \frac{\ET{\pi}}{\ET{\THRESH}}\geq1.$

\item If $\alpha = o\left(k/\tau\right)$, then $\kinfty \frac{\ET{\pi}}{\ET{\THRESH}}\geq1.$
\end{enumerate}
\end{restatable}
\end{samepage}
\begin{proof}
When $\mu_1 \geq \mu_2$, LPF, SERPT, and \thrpolicy are equivalent, and this result is immediate.  Hence, we will assume that $\mu_1 < \mu_2$.  We prove the three claims of Theorem \ref{thm:thr:opt} below.


\subsection*{Proof of Claim 1}
We first split the mean response time into two parts by conditioning on job class: 
$$\ET{\THRESH}= p_1\E[T_1^{\THRESH}]  + p_2\E[T_2^{\THRESH}].$$
To prove the asymptotic optimality of \thrpolicy, we begin by showing that $\lim_{k\to\infty}\E[T_1^{\THRESH}]=\frac{1}{\mu_1}$.

Our proof strategy is to prove that $\E[\text{\# class 1 jobs in queue}]=o(k)$.  Specifically, let $N_1$ and $K_1$ be the number of class 1 jobs in the system and in service, respectively, in steady state.  Then we know that $\E[\text{\# class 1 jobs in queue}] = \E[N_1] - \E[K_1] = \E[N_1] - k\rho_1$, where we have used the fact that $\E[K_1]=k\rho_1$, which can be verified using Little's law.  Then $\E[N_1] - k\rho_1$ can be further upper bounded as
\begin{align}
\E[N_1] - k\rho_1
&\le \E\left[(N_1-(k\rho_1+\delta))^+ + \delta\right],\label{eq:pospart}
\end{align}
where $\delta$ is chosen such that $k\rho_1+\delta$ is an integer and $\delta=\Theta(\sqrt{k})$, and the superscript $^+$ means taking the positive part, i.e., $x^+=\max\{x,0\}$ for any real number $x$.
We assume that $k$ is sufficiently large that $k\rho_1 + \delta < k$.
This inequality \eqref{eq:pospart} can be verified by considering the two cases that $N_1-(k\rho_1+\delta) \ge 0$ and $N_1-(k\rho_1+\delta) < 0$.
We choose to bound $\E[N_1] - k\rho_1$ in this particular way because $(N_1-(k\rho_1+\delta))^+$ effectively removes the dynamics of $N_1$ when $N_1 < k\rho_1+\delta$, enabling us to focus on the more tractable scenario where $N_1\ge k\rho_1+\delta$. 
Note that $\delta=\Theta(\sqrt{k})=o(k)$.
Therefore, it suffices to prove $\E\left[(N_1-(k\rho_1+\delta))^+\right]=o(k)$, which is the focus of the remainder of this proof.

We bound $\E\left[(N_1-(k\rho_1+\delta))^+\right]$ using Lyapunov drift analysis.
We consider the Lyapunov function $V(n_1,n_2)=\left((n_1-(k\rho_1+\delta))^+\right)^2$ for any state $(n_1,n_2)$.  Let $k_1$ denote the number of class 1 jobs in service when the state is $(n_1,n_2)$.
Then the drift of $V$ can be written as
\begin{align*}
\Delta V(n_1,n_2) &= \1{n_1=k\rho_1+\delta}\cdottight\lambdak_1+\1{n_1\ge k\rho_1+\delta+1}\cdot \Bigl(\lambdak_1\cdot\left(2(n_1-(k\rho_1+\delta))^+ +1\right)\\
&\mspace{270mu}+\mu_1k_1\left(-2(n_1-(k\rho_1+\delta))^+ +1\right)\Bigr).
\end{align*}
We know that in steady state, $\E[\Delta V(N_1,N_2)]=0$.  Therefore,
\begin{align}
&\mspace{21mu}\E\left[2(N_1-(k\rho_1+\delta))^+(\mu_1K_1-\lambdak_1)\cdot \1{N_1\ge k\rho_1+\delta+1}\right]\nonumber\\
&=\E\left[(\lambdak_1+\mu_1K_1)\cdot \1{N_1\ge k\rho_1+\delta+1}+\lambdak_1\cdot \1{N_1= k\rho_1+\delta}\right] \le \E\left[(\lambdak_1+\mu_1K_1)\cdot \1{N_1\ge k\rho_1+\delta}\right] \nonumber\\
&\le \lambdak_1+k\mu_1.\label{eq:drift-equality}
\end{align}

We lower bound the left-hand-side term $\E\left[2(N_1-(k\rho_1+\delta))^+(\mu_1K_1-\lambdak_1)\cdot \1{N_1\ge k\rho_1+\delta+1}\right]$.
We abuse the notation a bit and reuse $\mathrm{LPF}$ to also denote the event that the system is running LPF, i.e., the event that $N_1+N_2\le \tau$.
Then,
\begin{align}
&\mspace{21mu}\E\left[2(N_1-(k\rho_1+\delta))^+(\mu_1K_1-\lambdak_1)\cdottight \1{N_1\ge k\rho_1+\delta+1}\right]\nonumber\\
&=\E\left[2(N_1-(k\rho_1+\delta))^+(\mu_1K_1-\lambdak_1)\cdottight \1{\mathrm{LPF}\cap \{N_1\ge k\rho_1+\delta+1\}}\right]\nonumber\\
&\mspace{21mu}+\E\left[2(N_1-(k\rho_1+\delta))^+(\mu_1K_1-\lambdak_1)\cdottight \1{\mathrm{LPF}^c\cap \{N_1\ge k\rho_1+\delta+1\}}\right]\nonumber\\
&\ge \E\left[2(N_1-(k\rho_1+\delta))^+\cdottight\mu_1(\delta+1)\cdottight \1{\mathrm{LPF}\cap \{N_1\ge k\rho_1+\delta+1\}}\right]\label{eq:left1}\\
&\mspace{21mu}+\E\left[2(N_1-(k\rho_1+\delta))^+(-\lambdak_1)\cdottight \1{\mathrm{LPF}^c\cap \{N_1\ge k\rho_1+\delta+1\}}\right]\nonumber\\
&=\E\left[2(N_1-(k\rho_1+\delta))^+\cdottight\mu_1(\delta+1)\cdottight \1{\LPF}\right]+\E\left[2(N_1-(k\rho_1+\delta))^+(-\lambdak_1)\cdottight \1{\mathrm{LPF}^c\cap \{N_1\ge k\rho_1+\delta+1\}}\right]\label{eq:left2}\\
&=\E\left[2(N_1-(k\rho_1+\delta))^+\cdottight\mu_1(\delta+1)\right]-\E\left[2(N_1-(k\rho_1+\delta))^+\cdottight\mu_1(\delta+1)\cdottight \1{\mathrm{LPF}^c}\right]\nonumber\\
&\mspace{21mu}+\E\left[2(N_1-(k\rho_1+\delta))^+(-\lambdak_1)\cdottight \1{\mathrm{LPF}^c\cap \{N_1\ge k\rho_1+\delta+1\}}\right]\nonumber\\
&\ge \E\left[2(N_1-(k\rho_1+\delta))^+\cdottight\mu_1(\delta+1)\right]-\E\left[2(N_1-(k\rho_1+\delta))^+\cdottight(\mu_1(\delta+1)+\lambdak_1)\cdottight \1{\mathrm{LPF}^c}\right]\label{eq:left-lower},
\end{align}
where \eqref{eq:left1} is due to the fact that $K_1\ge k\rho_1+\delta+1$ when LPF is in use and $N_1\ge k\rho_1+\delta+1$, and \eqref{eq:left2} is due to the equality $(N_1-(k\rho_1+\delta))^+\cdot \1{N_1\ge k\rho_1+\delta+1}=(N_1-(k\rho_1+\delta))^+$.

We now plug the lower bound \eqref{eq:left-lower} back into \eqref{eq:drift-equality}, which gives
{\small
\begin{align}
\E\left[2(N_1-(k\rho_1+\delta))^+\cdottight\mu_1(\delta+1)\right]\le \lambdak_1+k\mu_1+\E\left[2(N_1-(k\rho_1+\delta))^+\cdottight(\mu_1(\delta+1)+\lambdak_1)\cdottight \1{\mathrm{LPF}^c}\right]\nonumber\\
\E\left[2(N_1-(k\rho_1+\delta))^+\right]\le\frac{\lambdak_1+k\mu_1}{2\mu_1(\delta+1)}+\frac{\E\left[(N_1-(k\rho_1+\delta))^+\cdot(\mu_1(\delta+1)+\lambdak_1)\cdot \1{\mathrm{LPF}^c}\right]}{\mu_1(\delta+1)}\nonumber.
\end{align}
}
Since the first term $\frac{\lambdak_1+k\mu_1}{2\mu_1(\delta+1)}=o(k)$, it suffices to show the second term is also $o(k)$.
To this end, we have
\begin{align}
\E\left[(N_1-(k\rho_1+\delta))^+\cdottight(\mu_1(\delta+1)+\lambdak_1)\cdottight \1{\mathrm{LPF}^c}\right]
&\le (\mu_1(\delta+1)+\lambdak_1)\cdottight \E\left[N_1\cdottight \1{\mathrm{LPF}^c}\right]\nonumber\\
&\le (\mu_1(\delta+1)+\lambdak_1)\cdottight \E\left[N_1^2\right]^{1/2}\cdottight (\Pr(\mathrm{LPF}^c))^{1/2}\label{eq:left-cs}\\
&\le (\mu_1(\delta+1)+\lambdak_1)\cdottight O(k^2)\cdottight (\Pr(\mathrm{LPF}^c))^{1/2}\label{eq:left-moment}.
\end{align}
Here, \eqref{eq:left-cs} follows from the Cauchy-Schwarz inequality, and \eqref{eq:left-moment} follows from the same moment bound used in Appendix \ref{sec:momentbound}.
What remains is to bound $\Pr(\mathrm{LPF}^c)$.

Recall that \thrpolicy performs LPF if $N_1+N_2\le \thr$.
Thus, $\Pr(\mathrm{LPF}^c)=\Pr(N_1+N_2 > \thr)$.
Given sufficiently large values of $k$, 
\begin{align}
N_1+N_2 > \thr  \Rightarrow \frac{N_1}{\mu_1}+\frac{N_2}{\mu_2} > \frac{N_1+N_2}{\mu_2} > \frac{\thr}{\mu_2} \label{eq:event1}
\Rightarrow \frac{N_1}{\mu_1}+\frac{N_2}{\mu_2} > \frac{2k}{\mu_1},
\end{align}
where \eqref{eq:event1} uses the assumption that $\mu_1<\mu_2$ and $\tau=\omega(k)$.
We consider the Lyapunov function $W(n_1,n_2)=\frac{n_1}{\mu_1}+\frac{n_2}{\mu_2}$.  When $W(n_1,n_2) \ge \frac{k}{\mu_1}$, we have either at least $k$ class 1 jobs or at least $1$ class 2 job in service.  This means all servers are busy under \thrpolicy.  Let $k_1$ and $k_2$ be the numbers of servers working on class 1 and class 2 jobs, respectively, in state $(n_1,n_2)$. Therefore, because $k_1 + k_2 = k$, we have
\begin{align*}
\Delta W(n_1,n_2) = \lambdak_1\frac{1}{\mu_1}+\lambdak_2\frac{1}{\mu_2}-k_1\cdot \mu_1\frac{1}{\mu_1}-k_2\cdot\mu_2\frac{1}{\mu_2}\nonumber&=k\rho_1+k\rho_2-k =-\alpha.
\end{align*}
Then, using the same tail bound from Appendix \ref{sec:driftbound}, we have
\begin{align}
\Pr(\frac{N_1}{\mu_1}+\frac{N_2}{\mu_2} > \frac{2k}{\mu_1})\le \left(\frac{(\lambdak_1+\lambdak_2)/\mu_1}{(\lambdak_1+\lambdak_2)/\mu_1+\alpha}\right)^{k+1}=O\left(e^{-d\alpha}\right),
\end{align}
for some constant $d>0$ independent of $k$.
Therefore,
\begin{align*}
\Pr(\mathrm{LPF}^c)=\Pr(N_1+N_2>\tau)=O\left(e^{-d\alpha}\right).
\end{align*}
Plugging this back into \eqref{eq:left-moment} yields
\begin{align}
\E\left[(N_1-(k\rho_1+\delta))^+\cdot(\mu_1(\delta+1)+\lambdak_1)\cdot \1{\mathrm{LPF}^c}\right]=O\left(k^3e^{-d\alpha}\right).
\end{align}

Putting everything together gives
\begin{align}
\E\left[2(N_1-(k\rho_1+\delta))^+\right]
&\le \frac{\lambdak_1+k\mu_1}{2\mu_1(\delta+1)}+\frac{\E\left[(N_1-(k\rho_1+\delta))^+\cdot(\mu_1(\delta+1)+\lambdak_1)\cdot \1{\mathrm{LPF}^c}\right]}{\mu_1(\delta+1)}\nonumber\\
&=o(k)+\frac{O\left(k^3e^{-c\alpha}\right)}{\mu_1(\delta+1)}\nonumber =o(k).
\end{align}

Applying Little's Law then gives
$$\kinfty \E[T_1] \leq \frac{o(k) + k\rho_1 + \delta}{\lambdak_1} = \frac{1}{\mu_1}.$$

To complete our proof, note that the response time of class 2 jobs under \thrpolicy should intuitively be lower than their response time under LPF, since the number of servers devoted to class 2 jobs under \thrpolicy is greater than or equal to the number of servers used for class 2 jobs by LPF in every state.
Formally, we can prove this using a drift argument that is nearly identical to that of Lemma \ref{if:shw:elastic:opt}.
Specifically, we pessimistically assume that class 2 jobs always get $(k-N_1)^+$ servers.
Then, we bound the number of class 1 jobs using the modified tail bound from Appendix \ref{sec:app:mod}.
Here, we use our bound on $Pr(\mathrm{LPF}^c)$ to bound the probability of having positive drift in the number of class 1 jobs.
From this point, the argument to bound $\E[N_2]$ is identical to the proof of Lemma \ref{if:shw:elastic:opt}, giving $\lim_{k\to\infty}\E[T_2^{\THRESH}]=0.$

Hence, the mean response time of both class 1 and class 2 jobs under THRESH is asymptotically optimal.
\subsection*{Proof of Claim 2}


Our proof follows the general argument of Theorem~\ref{thm:conventional} with some adjustments.
Recall that in the proof of Theorem~\ref{thm:conventional}, we consider the SERPT-1 system, which uses a single server of speed $k$ to run the optimal single-server policy that prioritizes jobs with the highest values of $\mu_i$.

For any scheduling policy $\pi$, $\lim_{\rho \rightarrow 1} \frac{\ET{\pi}}{\ET{\SERPTONE}}\geq1.$  Therefore, to prove that $\lim_{\rho \rightarrow 1} \frac{\ET{\pi}}{\ET{\THRESH}}\geq1$, it suffices to prove that $\lim_{\rho \rightarrow 1} \frac{\ET{\SERPTONE}}{\ET{\THRESH}}=1.$
Because $\lim_{\rho \rightarrow 1} \frac{\ET{\SERPTONE}}{\ET{\THRESH}}\le 1$,
it suffices to prove that $\lim_{\rho \rightarrow 1} \frac{\ET{\SERPTONE}}{\ET{\THRESH}}\ge 1.$

We first split the mean response time into two parts based on class 1 and class 2 jobs:
\begin{align}
\frac{\ET{\SERPTONE}}{\ET{\THRESH}}&=\frac{p_1\E[T_1^{\SERPTONE}] + p_2\E[T_2^{\SERPTONE}]}{p_1\E[T_1^{\THRESH}] + p_2\E[T_2^{\THRESH}]}.\nonumber
\end{align}

Intuitively, class 1 jobs should perform well under \thrpolicy compared to the SERPT-1 system, because class 1 jobs always have low priority in the SERPT-1 system.
Specifically, we can show that
\begin{equation}\E[N_1^{\THRESH}] \leq \E[N_1^{\SERPTONE}] + k.\label{eq:thbound}\end{equation}
To prove this, we compare \thrpolicy to a single-server system of speed $k$ that also runs \thrpolicy.
Using the same coupling argument from the proof of Theorem \ref{thm:conventional}, we show that the single-server version of \thrpolicy is within $k$ jobs of the $k$ server version of \thrpolicy on average.
Then, because the SERPT-1 system gives lower priority to class 1 jobs than \thrpolicy, the SERPT-1 system must have more class 1 jobs than the single-server \thrpolicy system on average.
Taken together, these arguments yield \eqref{eq:thbound}.
We then have
\begin{align}
\frac{\ET{\SERPTONE}}{\ET{\THRESH}} \geq \frac{p_1 \ET[1]{\SERPTONE} + p_2 \ET[2]{\SERPTONE}}{p_1\ET[1]{\SERPTONE} + O(1) + p_2\ET[2]{\THRESH}} &\geq \frac{p_1 \ET[1]{\SERPTONE}}{p_1\ET[1]{\SERPTONE} + O(1) + p_2\ET[2]{\THRESH}}\label{eq:cmu:th}.
\end{align}

Next, we know from classical queueing theory that 
\begin{equation}\E[T_1^{\SERPTONE}]\ge\frac{1}{1-\rho_2}\left(\frac{1}{k\mu_1}+\frac{\rho_2}{k\mu_2(1-\rho)}+\frac{\rho_1}{k\mu_1(1-\rho)}\right).\label{eq:th:bound}\end{equation}

We can also prove the following lemma, whose proof is given in Appendix \ref{sec:thresh:lems}.
\begin{restatable}{lemma}{tethresh}\label{lem:TETHRESH}
$\displaystyle\E[T_2^{\THRESH}]\le\frac{2\thr}{k\mu_2\rho_2}+\frac{4}{k\mu_2\rho_2(1-\rho_2)}.$
\end{restatable}
To finish our proof, we can apply these bounds to \eqref{eq:cmu:th}, giving
\begin{align*}
\frac{\ET{\SERPTONE}}{\ET{\THRESH}} &\geq \frac{p_1 \ET[1]{\SERPTONE}}{p_1\ET[1]{\SERPTONE} + O(1) + O(1)} \mand\\
\rhoone \frac{\ET{\SERPTONE}}{\ET{\THRESH}} &\geq \rhoone \frac{1}{1 + O(1-\rho)} = 1.
\end{align*}

\subsection*{Proof of Claim 3}
To prove claim 3, we can slightly modify our proof of claim 2.
Recall that for \thrpolicy in \snds regimes, we have assumed $\alpha=o(k/\tau)$.
Furthermore, we have that $\kinfty \rho_2=\rho_2^*$ where $\rho_2^*$ is a constant.
Therefore, in \snds regimes, the bound in \eqref{eq:th:bound} becomes
$\E[T_1^{\SERPTONE}]=\omega(\tau/k),$
and the upper bound in Lemma~\ref{lem:TETHRESH} becomes
$\E[T_2^{\THRESH}]=O(\tau/k).$
Applying these bounds to \eqref{eq:cmu:th} (which remains the same) gives
\begin{align*}
\frac{\ET{\SERPTONE}}{\ET{\THRESH}} &\geq \frac{p_1 \ET[1]{\SERPTONE}}{p_1\ET[1]{\SERPTONE} + O(1) + O\left(\frac{\tau}{k}\right)} \mand\\
\kinfty \frac{\ET{\SERPTONE}}{\ET{\THRESH}} &\geq \kinfty \frac{1}{1 + o\left(\frac{k}{\tau}\right) + o(1)} = 1.
\end{align*}
Taken together, these three claims give Theorem \ref{thm:thr:opt}.
\end{proof}



\section{Evaluation}\label{sec:eval}

In this section, we evaluate our theoretical results via simulation.
While our theoretical results require some technical assumptions, we find that the high-level lessons from these theorems ---  deferring parallelizable work in lighter load and prioritizing short jobs in heavier load ---  apply more generally.

\subsection{\thrpolicy with Additional Classes}
While we prove the asymptotic optimality of \thrpolicy in a special case of our model with $\ell=2$ job classes, we conjecture that \thrpolicy works well in general for any number of job classes with arbitrary levels of parallelizability.
To validate this intuition, we simulate LPF, SERPT, and THRESH in various scaling regimes in Figure \ref{fig:sims}.
Indeed, given $\ell=4$ job classes with a variety of parallelizability levels, THRESH performs near-optimally in both \shw and \snds scaling regimes.

While our theoretical results only consider the \shw and \snds scaling regimes, we simulated the performance of these policies in a middle scaling regime that lies between these cases for the sake of completeness.  These results are presented in Appendix \ref{sec:app:middle}. 

\begin{figure}
\centering
\includegraphics[width=.95\textwidth]{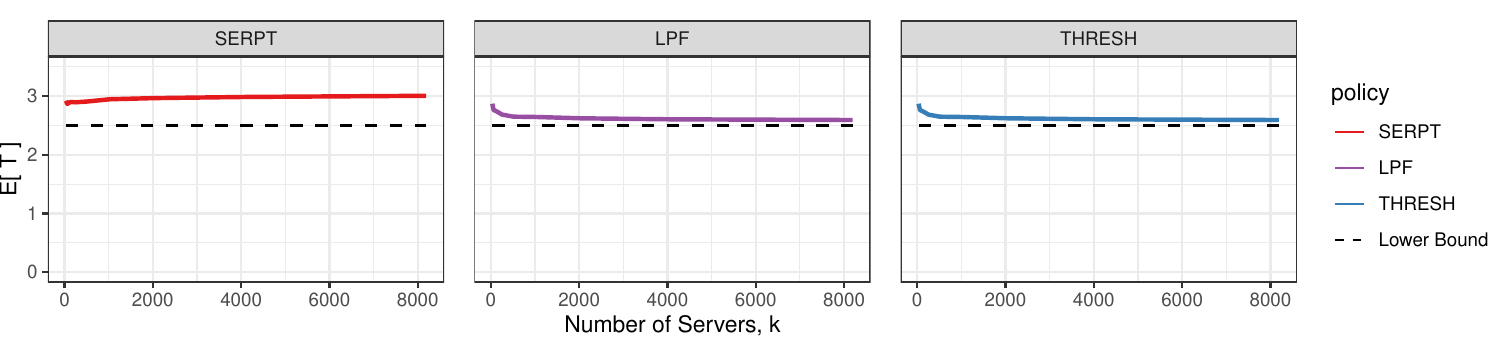}
\includegraphics[width=.95\textwidth]{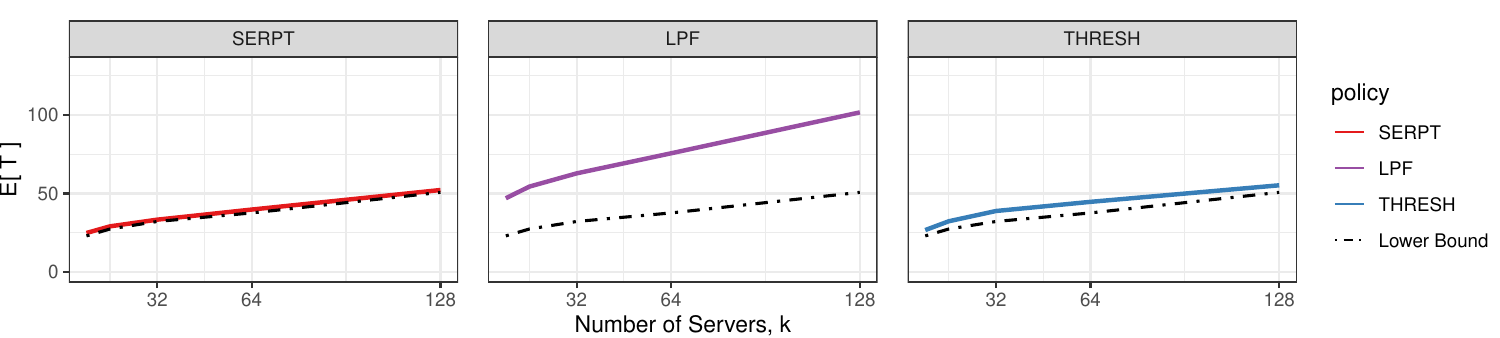}
\caption{Simulations of various scheduling policies in \shw (top) and \snds (bottom) regimes with $\alpha=2k^{3/4}$ and $\alpha=2k^{-1/4}$ respectively.  Here, $\ell=4$, and $c_1=1$, $c_2=4$, $c_3=\log k$, $c_4=k$ and $\mu_1=.2$, $\mu_2=.05$, $\mu_3=.3$, $\mu_4=.1$.  
In \shw regimes, mean response time under LPF converges to mean service time (our ``lower bound'').
In \snds regimes, mean response time under SERPT is close to that of the SERPT-1 policy, a lower bound that assumes a single fast server.
\thrpolicy performs well in both regimes.
}
\label{fig:sims}
\end{figure}

\subsection{General Job Size Distributions}\label{sec:eval:gen}
We have thus far assumed that each class of jobs, $i$, is associated with some exponential job size distribution.
It is not clear how our results generalize when job sizes are generally distributed.

First, consider scheduling generally distributed parallelizable jobs in conventional heavy traffic.
Recent results \cite{scully2020gittins} have demonstrated that a Gittins index policy is optimal for scheduling non-parallelizable jobs in conventional heavy traffic.
We find that these results hold when jobs are parallelizable.
Specifically, \cite{scully2020gittins} defines a quantity called $B(r)$ that tracks the number of servers used to run jobs with a Gittins index of $r$ or less.
When moving from non-parallelizable jobs to parallelizable jobs, $B(r)$ does not decrease for any given system state, implying that the bounds from this work continue to hold when scheduling parallelizable jobs.
The key fact exploited in \cite{scully2020gittins} is that, when $B(r)$ is small because fewer than $k$ servers are working on a job with Gittins index at most $r$, there must be at most $k-1$ rank $r$ jobs in the system.
When jobs are parallelizable, this same principle holds since we now assume that each job runs on \emph{at least} one server.
Following the arguments from \cite{scully2020gittins}, we find that prioritizing parallelizable jobs based on their Gittins indices is optimal in conventional heavy traffic.

Unfortunately, this argument does not generalize cleanly to other scaling regimes.
The proof of heavy-traffic optimality in \cite{scully2020gittins} relies on the analysis of SRPT in conventional heavy traffic from \cite{lin2010average}.
Currently, the behavior of SRPT in other scaling regimes is not well-understood.
We conjecture that a Gittins-based policy is asymptotically optimal in \snds scaling regimes, but a proof of this claim does not follow immediately from \cite{scully2020gittins} and is beyond the scope of this paper.
\begin{figure}
\includegraphics[width=.95\textwidth]{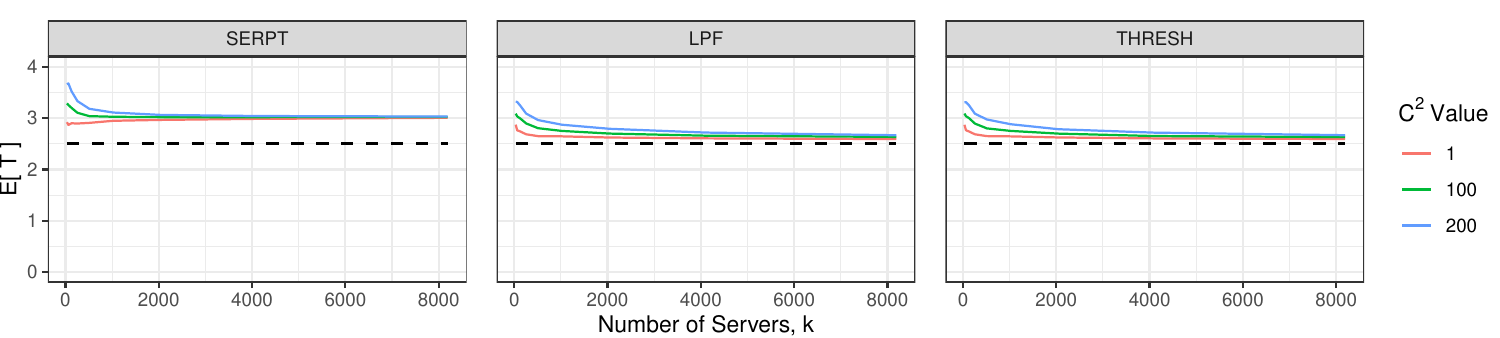}
\caption{Simulations showing the effect of variability on mean response time.  Experiments use the same $\ell=4$ job classes as Figure \ref{fig:sims}, but with different job size distributions.  When $C^2=1$, the job size distribution is exponential.  As $C^2$ increases, we examine hyperexponential job size distributions with two states such that the mean job size remains the same, but the squared coefficient of variation is increased.
}
\label{fig:hyper}
\end{figure}

Second, consider scheduling non-exponentially-distributed jobs in \shw scaling regimes.
Given that our results hold for $\ell$ exponential job classes, we conjecture that our results hold for phase-type job size distributions.
Specifically, our current proof of the optimality of LPF bounds $\E[T_i]$ by showing that the departure rate of class-$i$ jobs is almost always greater than $\lambdak_i$.
Given jobs with phase-type distributions, we can use a similar argument to bound the time required to complete each phase of each type of job.
However, when considering some phase, $j$, of class $i$ jobs, the rate at which jobs enter phase $j$ is not stationary --- it depends on the number of class $i$ jobs in service whose next phase could be $j$.
Hence, we must argue that the completion rate of jobs in phase $j$ is almost always greater than the time-varying rate at which jobs enter phase-$j$.
We can easily compute the long-run average rate at which jobs enter phase $j$, but we must show that the fluctuations around this average rate do not cause problems for a policy like LPF.

While a proof for the case of phase-type job sizes is beyond the scope of this paper, we check via simulation to confirm that our conjecture is plausible.
Figure \ref{fig:hyper} shows that our results appear to hold given jobs that follow different hyperexponential distributions with various levels of job size variability.
Here, $C^2=1$ corresponds to the same exponential distributions as Figure \ref{fig:sims}.
The larger values of $C^2$ denote that the job size distributions are balanced-mean hyperexponential distributions with the same mean job sizes but correspondingly higher squared coefficients of variation.
We find that, while convergence appears slower as $C^2$ increases, LPF and THRESH appear to be asymptotically optimal in \shw scaling regimes.

\section{Conclusion}
This paper presents the first theoretical results on balancing the tradeoff between prioritizing short jobs and deferring parallelizable work.
Additionally, this paper provides the first optimality results on scheduling more than two classes of parallelizable jobs, where each class has its own job size distribution and parallelizability level.
We show that the choice of an asymptotically optimal scheduling policy depends on the scaling regime chosen.
LPF is asymptotically optimal in \shw regimes, and SERPT is asymptotically optimal in \snds regimes.
The THRESH policy, which combines LPF and SERPT, is asymptotically optimal in both \shw and \snds scaling regimes.

We see two main avenues for extending our results.
First, as discussed in Section \ref{sec:eval:gen}, allowing jobs to follow arbitrary size distributions complicates our analysis.
Existing techniques for analyzing $M/G/k$ systems with non-parallelizable jobs have only considered scheduling in conventional heavy traffic, and it remains unclear how to extend these results both to lighter-load scaling regimes and to scheduling parallelizable jobs.
Second, we have assumed that our jobs are perfectly parallelizable up to some parallelizability level.
In practice, however, jobs may receive a \emph{sublinear} speedup from parallelization.
Hence, while our results consider a work-conserving system, it is of great practical interest to analyze a non-work-conserving variant of our problem.
We note that this generalization is related to the problem of state-dependent service rates \cite{ayesta2020scheduling} and the analysis of systems with Markov-Modulated service rates \cite{duran2022whittle}.
We hope that the techniques and intuition developed in these papers can be applied to the problem of scheduling jobs with sublinear speedups.

{\small
     \bibliographystyle{abbrvnat}
     \bibliography{bibshort}
}
\begin{appendices}
\section{Asymptotic Notation}\label{sec:asymptoticnote}
We use standard asymptotic notation throughout the paper to describe scaling behavior.
Specifically, consider two functions $f(k), g(k) > 0$, and let $c_1, c_2 >0$ be arbitrary constants.  We define the relevant notation in the table below.
\begin{table}[h]
\label{tab:asymptotes}
\resizebox{\textwidth}{!}{%
\begin{tabular}{|c|c|c|c|}
\hline
Name & Notation & Definition & Description \\
\hline
  Little-o   &  $f(k)=o(g(k))$     &   $\kinfty \frac{f(k)}{g(k)} = 0$ & $f(k)$ grows much slower than $g(k)$        \\
  \hline
  Big-O   &  $f(k)=O(g(k))$      &     $\limsup_{k\to\infty} \frac{f(k)}{g(k)} \leq c_1$ & $f(k)$ grows weakly slower than $g(k)$       \\
  \hline
  Big-$\Theta$   &    $f(k)=\Theta(g(k))$     &    $c_1 \leq \limsup_{k\to\infty} \frac{f(k)}{g(k)} \leq c_2$ & $f(k)$ grows roughly as fast as $g(k)$        \\
  \hline
  Little-$\omega$ &  $f(k)=\omega(g(k))$      & $\kinfty \frac{f(k)}{g(k)} = \infty$ & $f(k)$ grows much faster than $g(k)$\\
  \hline
\end{tabular}%
}
\end{table}

%
%
%

\section{A Drift-Based Tail Bound}\label{sec:driftbound}

We use Lemma 10 from \cite{wang2022heavy} to bound the tail probability of $N_1^{\LPF}$.
This lemma defines a Lyapunov function, $V(x)$.  
If the drift of the Lyapunov function is negative when $V(x) > y$, the lemma bounds the probability that $V(X) \gg y$.
A full statement of Lemma 10 from \cite{wang2022heavy} can be found in Appendix~\ref{ecomp}.
Using this lemma, we prove the following tail bound.

\begin{restatable}{lemma}{driftbound}\label{lem:driftbound}
For some constant $d > 0$,
\begin{equation*}
Pr\left(N_1^{\LPF} \geq \rho_1k + \frac{\alpha}{\ell}\right) = O(e^{-d\alpha^2/k})
\end{equation*}

\end{restatable}
\begin{proof}
We begin by defining a Lyapunov function $V(x)=n_1$.
Let $\mu_{max}=\max_{1\leq i \leq \ell} \mu_i$.

When $n_1 \geq \rho_1 k + \frac{\alpha}{2\ell}$, the drift of the Lyapunov function, $\Delta V$, is
$$\Delta V(x) = \lambdak_1 - \mu_1 \cdot n_1 \leq \lambdak_1 - \mu_1 (\rho_1 k + \frac{\alpha}{2\ell}) = -\mu_1 \frac{\alpha}{2\ell}.$$

The other conditions needed to apply the lemma from \cite{wang2022heavy} are trivially satisfied.
Furthermore, the maximum rate of departures from any state, $q_{max}$, is $\mu_{max} + \lambdak$.
We therefore obtain the following bound:
$$Pr\left(N^{\LPF}_1 \geq \left(\rho_1 k + \frac{\alpha}{2\ell}\right) + 2\cdot\frac{\alpha}{4\ell}\right) \leq \left(\frac{q_{max}}{q_{max} + \frac{\alpha}{2\ell}}\right)^{1+\frac{\alpha}{4\ell}} \leq  \left(\frac{q_{max}}{q_{max} + \frac{\alpha}{2\ell}}\right)^{\frac{\alpha}{4\ell}}$$

We then simplify this expression as follows,
\begin{align*}
Pr\left(N^{\LPF}_1 \geq \left(\rho_1 k + \frac{\alpha}{2\ell}\right) + 2\cdot\frac{\alpha}{4\ell}\right) &\leq \left(\frac{q_{max}+\frac{\alpha}{2\ell}}{q_{max} + \frac{\alpha}{2\ell}} - \frac{\frac{\alpha}{2\ell}}{q_{max} + \frac{\alpha}{2\ell}} \right)^{\frac{\alpha}{4\ell}}\\
&= \left(1-\frac{1}{\frac{2\ell \cdot q_{max}}{\mu_1\alpha}+1}\right)^{\frac{\alpha}{4\ell}}\\
&= \left(1-\frac{1}{\frac{k}{\mu_1\alpha}(2\ell\mu_{max} + \frac{\lambdak}{k}+\frac{\alpha}{k})}\right)^{\frac{\alpha}{4\ell}}\\
&\leq \left(1-\frac{d}{\frac{k}{\alpha}}\right)^{\frac{\alpha}{4\ell}}
\end{align*}
for some constant $d>0$.  Note that this constant $d$ exists because $\alpha < k$ and $\lambdak = \Theta(k)$.
To get the desired bound, we can then write
\begin{align*}
Pr\left(N^{\LPF}_1 \geq \left(\rho_1 k + \frac{\alpha}{2\ell}\right) + 2\cdot\frac{\alpha}{4\ell}\right) &\leq \left(\left(1-\frac{c_1}{\frac{k}{\alpha}}\right)^{\frac{k}{\alpha}}\right)^{\frac{\alpha^2}{4\ell k}}.
\end{align*}
For any positive value of $k$, this yields
\begin{align*}
Pr\left(N^{\LPF}_1 \geq \left(\rho_1 k + \frac{\alpha}{2\ell}\right) + 2\cdot\frac{\alpha}{4\ell}\right) = Pr\left(N^{\LPF}_1 \geq \rho_1 k + \frac{\alpha}{\ell}\right)  &\leq e^{-d\alpha^2/4\ell k}.
\end{align*}
\end{proof}

\section{A Modified Drift-Based Tail Bound}\label{sec:app:mod}

We now apply a slightly modified version of the drift bound from Appendix \ref{sec:driftbound} to bound $N_i^{LPF}$.
Whereas the original bound requires strictly negative drift in the Lyapunov function when its value is sufficiently large, the modified tail bound allows positive drift in the Lyapunov function as long as the stationary probability of being in the states with positive drift is bounded.
This bound is originally proven as Lemma~A.1 of \cite{weng2020achieving} (see the e-companion for details).

\begin{restatable}{lemma}{lem-mod-bound}\label{lem:mod:bound}
For any class of jobs, $i+1 \leq I$, there exist constants $\const{f}{i+1}{i+1}>0$ and $\const{d}{i+1}{i+1}>0$ such that
$$Pr(N^{\LPF}_{i+1} \geq \beta_{i+1}) \leq \const{f}{i+1}{i+1}e^{-\const{d}{i+1}{i+1}\alpha^2/k} + O\left(\frac{k}{\alpha}\right)Pr(X \notin \Eps_i)$$
where $\Eps_i$ is the set of states such that $n_j < \beta_j$ for all $j\leq i$.
\end{restatable}
\begin{proof}
We define a Lyapunov function $V(x)=n_{i+1}$.

For any state, $x$, in $\Eps_i$, the number of servers available to class $i+1$ jobs is at least
$$k-\sum_{j=1}^i c_j \cdot\beta_j \geq k - \sum_{j=1}^i \rho_jk +\frac{\alpha}{\ell} +c_j \geq \rho_{i+1}k + \frac{\alpha}{\ell} - \sum_{j=1}^i c_j.$$
Hence, as long as a state is in $\Eps_i$, we can apply the same drift argument used in \ref{lem:driftbound}.
Here, we are relying on the assumption that, for any $i\leq I$,
$$\sum_{j=1}^{i} c_j = O(1) \mand \alpha=\omega(1),$$
so the extra servers used by the lower-class jobs do not affect our drift argument order-wise.

The issue is that when a state is not in $\Eps_i$, there may not be enough free servers to guarantee negative drift of the Lyapunov function.  However, the positive drift in these states is upper bounded by $\lambdak_{i+1} = O(k)$.  Hence, according to the lemma from \cite{weng2020achieving}, we can bound the tail probability of $N_{i+1}$ in terms of the stationary probability that a state lies outside $\Eps_i$ and causes positive drift in our Lyapunov function.  Specifically, we have
$$Pr(N^{\LPF}_{i+1} > \beta_{i+1}) \leq \const{f}{i+1}{i+1}e^{-\const{d}{i+1}{i+1}\alpha^2/k} + O\left(\frac{k}{\alpha}\right)Pr(X \notin \Eps_i)$$
for some constants $\const{f}{i+1}{i+1}>0$ and $\const{d}{i+1}{i+1}>0$.
\end{proof}

\section{Proof of Equation (\ref{eq:product})}\label{sec:momentbound}
We divide the proof of \eqref{eq:product} into two parts: a technical lemma and a main lemma.

We begin by proving the following technical lemma using the drift-based moment bound from Lemma 10 of \cite{wang2022heavy}.
See Appendix~\ref{ecomp} for a full statement of the moment bound.
\begin{restatable}{lemma}{momentbound}\label{lem:momentbound}
For any class of jobs, $i$,
$$\E\left[\left(N_i^{\LPF}\right)^2\right] = O\left(k^2 + \left(\frac{k}{\alpha}\right)^2\right).$$
and thus in the \shw regime,
$$\E\left[\left(N_i^{\LPF}\right)^2\right] = O(k^2).$$
\end{restatable}
\begin{proof}

We begin by defining the Lyapunov function $V(x) = \sum_{j=1}^{i} \frac{n_i}{\mu_i}$ and let $\mu_{min} = \min_{j\leq i} \mu_j$.
Consider any state $x$ where $V(x) \geq \frac{2k}{\mu_{min}}$.  This implies that
$\sum_{j=1}^i n_i \frac{\mu_{min}}{\mu_i} \geq 2k$
and thus
$\sum_{j=1}^i n_i  \geq 2k.$

Because there are more than $k$ jobs in the system, and each job can use at least one server, there are no idle servers when $V(x) \geq \frac{2k}{\mu_{min}}$.  Let $k_i$ be the number of servers allocated to class $i$ jobs when the system is in state $x$.  We have
\begin{align}
\Delta V(x) &= \sum_{j=1}^i \frac{\lambdak_j}{\mu_j} - \sum_{j=1}^i \frac{k_j\mu_j}{\mu_j}\\
&=-k(1-\sum_{j=1}^i\rho_i) = -\alpha.
\end{align}

The total transition rate out of any state is upper bounded by some $O(k)$ term.
Furthermore, the change in $V$ between adjacent states is bounded by $\sum_{j=1}^\ell\frac{1}{\mu_j}$.
Hence, the bound from \cite{wang2022heavy} yields
$$\E\left[V(X)^2\right] \leq d_1k^2 + d_2 + \left(\frac{d_3 k + \alpha}{\alpha}\right)^2$$
for some constants $d_1, d_2, d_3 >0$.
This gives
$$\E\left[\left(N_i^{\LPF}\right)^2\right] \leq \mu_i \cdot \E\left[V(X)^2\right] = O\left(k^2 +\left(\frac{k}{\alpha}\right)^2\right).$$

In the \shw regime, we have $\alpha=\omega(\sqrt{k\log{k}})$, and thus
$$\E\left[\left(N_i^{\LPF}\right)^2\right] = O(k^2).$$
\end{proof}

We now use the technical lemma to prove our main lemma, restated from \eqref{eq:product}.
\begin{restatable}{lemma}{lem-lb}\label{lem:lb}
For sufficiently large $k$, when $\alpha = \omega\left(\sqrt{k\log k}\right)$,
$$\E[K_i\delta_i] \geq (\rho_i k + \log k)\E[\delta_i] - O(1), \quad \forall i\leq I.$$ 
\end{restatable}
\begin{proof}
To get this lower bound, we will use the bound on $Pr(X\notin \Eps_{i-i})$ from Lemma \ref{lem:shw:inelastic}.
Note that when $X\in \Eps_{i-1}$, we have
$$K_i \geq k-\sum_{j=1}^{i-1}k\rho_j + \frac{\alpha}{\ell} + c_j \geq k\rho_i + \frac{\alpha}{\ell}- \sum_{j=i}^{i-1} c_j,$$
and thus, for any $i\leq I$ and for sufficiently large $k$,
\begin{align}
\E[K_i\delta_i] \geq \E[\1{N_i\geq\theta_i}\1{X \in \Eps_{i-1}}\delta_iK_i] & \geq (k\rho_i + \log k )\E[\1{N_i \geq \theta_i}\1{X\in\Eps_{i-1}}\delta_i]\label{eq:region}\\
&\geq (k\rho_i+\log k)\E[\1{N_i \geq \theta_i}\delta_i - \1{\delta_i \geq \theta_i}\1{X\notin\Eps_{i-1}}\delta_i]\\
&\geq (k\rho_i+\log k )\E[\1{N_i \geq \theta_i}\delta_i - \1{X\notin\Eps_{i-1}}\delta_i]\\
&= (k\rho_i+\log k )\E[\delta_i - \1{X\notin\Eps_{i-1}}\delta_i].\label{eq:cauchy}
\end{align}
Here, we exploit the fact that there are $\rho_ik + \frac{\alpha}{\ell} - O(1)$ free cores when $X\in \Eps_{i-1}$, and we are conditioning on having enough jobs to use $\rho_ik + \log k$ cores.  Because we are in a \shw scaling regime, $\frac{\alpha}{\ell} - O(1) > \log k$ for sufficiently large $k$, so \eqref{eq:region} holds for sufficiently large $k$.

We now apply Cauchy-Schwarz to show that the second, negative term inside the expectation is $o\left(\frac{1}{k}\right).$
That is, we want to show that
$$\E[\1{X\notin\Eps_{i-1}}\delta_i] \leq \sqrt{\E[(\1{X\notin\Eps_{i-1}})^2]\E[\delta_i^2]} = \sqrt{\E[(\1{X\notin\Eps_{i-1}})]\E[\delta_i^2]} = o\left(\frac{1}{k}\right).$$
For the indicator term, it suffices to apply the bound from Lemma \ref{lem:shw:inelastic}.
We then need to show that $\E[\delta_i^2]$ is not too large.
We bound $\E[\delta_i^2]$ using our technical lemma, Lemma \ref{lem:momentbound}, to get
$$\E[\delta_i^2] \leq \E[N_i^2] \leq \mu_iE[V(x)^2] =O(k^2).$$
Hence, the second term in \eqref{eq:cauchy} is $o(\frac{1}{k})$. 
This gives the desired bound on $\E[K_i\delta_i]$.
\end{proof}

\section{Proof of Lemma \ref{lem:TETHRESH}}\label{sec:thresh:lems}

\tethresh*
\begin{proof}
Consider the Lyapunov function $V(n_1,n_2)=n_2$ for any state $(n_1,n_2)$.  Then, when $V(n_1,n_2)\ge \thr$, we know that \thrpolicy must be prioritizing class 2 jobs, and thus the drift is given by
\begin{align*}
\Delta V(n_1,n_2) = \lambdak_2 - k\mu_2 = - k\mu_2(1-\rho_2).
\end{align*}
Therefore, by the moment bounds in \cite{wang2022heavy}, we have
\begin{align*}
\E[N_2^{\THRESH}]&\le 2\thr + 4\cdot\frac{\lambdak_2+k\mu_2(1-\rho_2)}{k\mu_2(1-\rho_2)}\\
&=2\thr + \frac{4}{1-\rho_2}.
\end{align*}
Then by Little's law, we have
\[
\E[T_2^{\THRESH}]=\frac{\E[N_2^{\THRESH}]}{\lambdak_2}\le\frac{2\thr}{k\mu_2\rho_2}+\frac{4}{k\mu_2\rho_2(1-\rho_2)}.
\]
\end{proof}

\section{Simulations in a Middle Scaling Regime}
While our theoretical results apply to both the \shw and \snds scaling regimes, we can simulate the performance of various policies in the middle scaling regimes between these two cases.
Figure \ref{fig:middle} shows that performance in the middle scaling regimes appears similar to performance in the \shw regimes, with LPF and \thrpolicy converging to the optimal mean response time as $k\to\infty$.
\label{sec:app:middle}
\begin{figure}[ht]
\centering
\includegraphics[width=.95\textwidth]{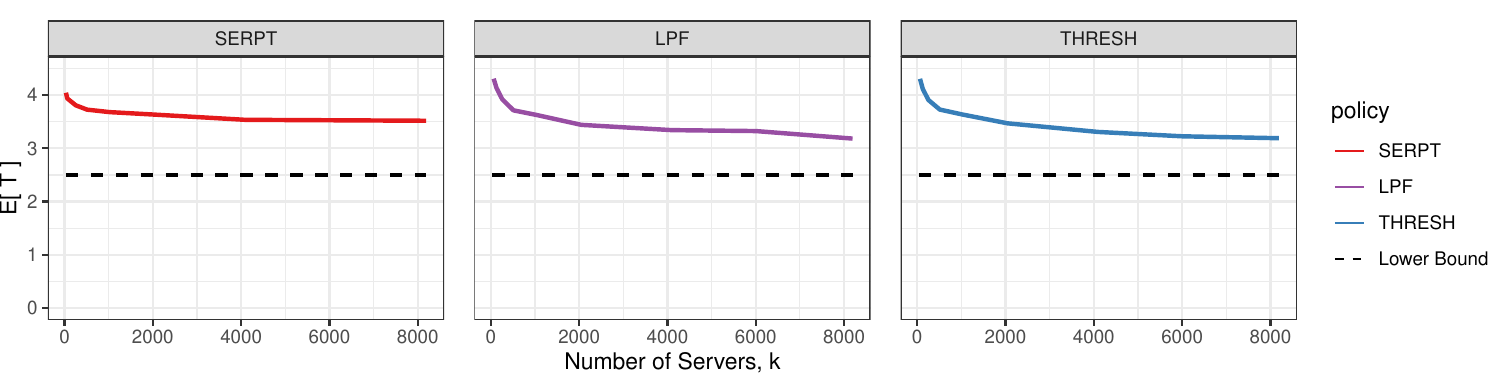}
\caption{Simulations of various scheduling policies in a middle scaling regime with $\alpha=2k^{1/4}$.  Here, $\ell=4$, and $c_1=1$, $c_2=4$, $c_3=\log k$, $c_4=k$ and $\mu_1=.2$, $\mu_2=.05$, $\mu_3=.3$, $\mu_4=.1$.  
In this case, as in the \shw regimes, both LPF and \thrpolicy appear to perform well as $k$ increases while SERPT appears to converge to an asymptotically suboptimal mean response time.
}
\label{fig:middle}
\end{figure}

\section{Full Statement of Drift Bounds}\label{ecomp}
\subsection{Lemma 10 from \cite{wang2022heavy}}
\begin{lemma}[Lemma 10 from \cite{wang2022heavy}]\label{rmk:tail}
Let $\{X(t)\}$ for $t\geq0$ be an irreducible, nonexplosive, positive-recurrent continuous-time Markov chain over a countable state space $\chi$.  Let $X$ be the stationary distribution of this Markov chain.  Consider a Lyapunov function $V: \chi \to \mathbb{R}_+$.  Let $\Delta V(x)$ be the drift of the Lyapunov function, defined as
$$\Delta V(x) = \sum_{x'\in \chi: x \neq x'} q_{xx'}\left(V(x')-V(x)\right)$$
where $q_{xx'}$ is the corresponding entry of the infinitesimal generator of $\{X(t)\}$.
If $V$ satisfies the following conditions
\begin{enumerate}
\item There exist constants $\gamma > 0$ and $B>0$ such that $V(x) > B \implies \Delta V(x) \leq -\gamma$ for any $x\in\chi$.
\item There exists a constant $\nu_{max} < \infty$ such that $\sup_{x, x' \in \chi : q_{xx'} >0} \mid V(x') - V(x) \mid = \nu_{max}$.
\item There exists a constant $\overline{q} < \infty$ such that $\sup_{x\in\chi} -q_{xx} = \overline{q}$.
\end{enumerate}
Then for any $j>0$, we have
\begin{equation}
Pr\left(V(X) > B + 2\nu_{max}j\right) \leq \left(\frac{q_{max}\nu_{max}}{q_{max}\nu_{max} + \gamma}\right)^{j+1}\label{eq:tail:weina}
\end{equation}
where
$$q_{max} = \sup_{x\in\chi} \sum_{x'\in\chi: V(x) < V(x')} q_{xx'}.$$
The bound in \eqref{eq:tail:weina} can then be used to bound the $m$th moment of $V(X)$ as follows:
\begin{equation}
\E[V(X)^m] \leq (2B)^m + (4\nu_{max})^m\left(\frac{q_{max}\nu_{max} + \gamma}{\gamma}\right)^m m!\label{eq:moment:weina}
\end{equation}
\end{lemma}

\subsection{Lemma A.1 from \cite{weng2020achieving}}
\begin{lemma}[Lemma A.1 from \cite{weng2020achieving}]
This lemma extends the result from Lemma \ref{rmk:tail}.
While this lemma was originally proven for finite-state continuous-time Markov chains, it is trivially extended to infinite-state Markov chains obeying the conditions in Lemma \ref{rmk:tail}.
This lemma adds the condition that some states with $V(x) > B$ are permitted to have bounded positive drift.
Specifically, for constants $\delta >0$, $\gamma>0$, and $B>0$, let $\Eps$ be a set of states such that:
\begin{itemize}
\item $V(x) > B \implies \Delta V(x) \leq -\gamma$ for any $x\in\Eps$.
\item $V(x) > B \implies \Delta V(x) \leq \delta$ for any $x\notin\Eps$.
\end{itemize}
In this case, the tail bound in \eqref{eq:tail:weina} becomes
\begin{equation}
Pr\left(V(X) > B + 2\nu_{max}j\right) \leq \left(\frac{q_{max}\nu_{max}}{q_{max}\nu_{max} + \gamma}\right)^{j} +\left(\frac{\delta}{\gamma} +1 \right)Pr(X\notin\Eps) \label{eq:tail:mod:weina}
\end{equation}
\end{lemma}
\end{appendices}

\end{document}